\documentclass[french, 11pt, a4paper, oneside,bibliography=totocnumbered]{scrartcl}
\usepackage[T1]{fontenc} 
\usepackage[latin1]{inputenc} 
\usepackage{mathptmx} 
\usepackage{lmodern} 
\usepackage{amsthm} 
\usepackage{amsmath} 
\usepackage{amsfonts} 
\usepackage{amssymb} 
\usepackage{mathrsfs}

\usepackage{bbm}
\usepackage[left=2.2cm, right=2.2cm, top=2.3cm, bottom=2.3cm]{geometry} 
\usepackage[ngerman,german,english]{babel} 
\usepackage{paralist} 
\usepackage[colorlinks=true,citecolor=blue]{hyperref}
\hypersetup{linktocpage}
\usepackage{etoolbox}

\usepackage[nottoc,notlot,notlof]{tocbibind}
\usepackage{graphicx}
\usepackage{xcolor}

\setkomafont{sectionentry}{\normalsize}
\RedeclareSectionCommand[tocbeforeskip=1pt]{section}

\setkomafont{section}{\normalfont\Large\textbf}
\setkomafont{subsection}{\normalfont\large\textbf}
\setkomafont{paragraph}{\normalfont\textbf}

\renewenvironment{abstract}
{\begin{center}
		\textbf{Abstract}
	\end{center}
	\list{}{ 
		\setlength{\leftmargin}{0.05\textwidth}
		\setlength{\rightmargin}{\leftmargin}
	}
	\item\relax} 
{\endlist}

\newenvironment{keywords}
{\begin{trivlist}\item[]{\bfseries Keywords.}}
	{\end{trivlist}}

\numberwithin{equation}{section}
\theoremstyle{plain} 
\newtheorem{theorem}{Theorem}[section] 
\newtheorem{corollary}[theorem]{Corollary}
\newtheorem{lemma}[theorem]{Lemma}
\newtheorem{proposition}[theorem]{Proposition}

\newtheorem{definition}[theorem]{Definition}
\newtheorem{remark}[theorem]{Remark}
\newtheorem{example}[theorem]{Example}

\newcommand{\N}{\mathbb{N}}

\newcommand{\R}{\mathbb{R}}

\newcommand\restr[2]{{
		\left.\kern-\nulldelimiterspace 
		#1 
		\vphantom{\big|} 
		\right|_{#2} 
}}

\newcommand{\norm}[2][]{\left\|#2\right\|_{#1}}

\newcommand{\Lpnorm}[2][]{\ifthenelse{\equal{#1}{}}{\norm{#2}_{L^p}}{\norm{#2}_{L^p(#1)}}}
\newcommand{\Hknorm}[2][]{\ifthenelse{\equal{#1}{}}{\norm{#2}_{H^k}}{\norm{#2}_{H^k(#1)}}}

\newcommand{\setdef}[2]{\left\lbrace #1 \ : \ #2 \right\rbrace}

\renewcommand{\ker}{\text{\normalfont ker\,}}

\newcommand{\QV}[2][]{\ifthenelse{\equal{#1}{}}{\langle #1 \rangle}{\langle #1,#2 \rangle}}

\renewcommand{\P}{\mathbb{P}}

\newcommand{\ind}{\mathbbm{1}}

\makeatletter
\newcounter{author}
\renewcommand*\author[1]{%
	\stepcounter{author}%
	\ifnum\c@author=1
	\gdef\@author{#1}%
	\else
	\xdef\@author{\unexpanded\expandafter{\@author\and#1}}%
	\fi
	\csgdef{author@\the\c@author}{#1}}
\newcommand*\email[1]{%
	\csgdef{email@\the\c@author}{#1}}
\newcommand*\address[1]{%
	\csgdef{address@\the\c@author}{#1}}
\AtEndDocument{%
	\xdef\author@count{\the\c@author}%
	\c@author=1
	\print@authors}
\newcommand*\print@authors{%
	\ifnum\c@author>\author@count
	\else
	\print@author{\the\c@author}%
	\advance\c@author by 1
	\expandafter\print@authors
	\fi}
\newcommand*\print@author[1]{%
	\par\medskip
	\begin{tabular}{@{}l@{}}%
		\textsc{\csuse{author@#1}}\\
		\csuse{address@#1}\\
		\textit{E-mail address}: \csuse{email@#1}
\end{tabular}}
\makeatother

\title{\textrm{\textbf{\Large Characterizing the detailed balance property by means of measurements in chemical networks}}}

\author{\large Eugenia Franco}
\address{University of Bonn, Institute for Applied Mathematics \\ Endenicher Allee 60 \\ D-53115 Bonn \\ GERMANY}
\email{franco@iam.uni-bonn.de}

\author{\large Bernhard Kepka}
\address{University of Bonn, Institute for Applied Mathematics \\ Endenicher Allee 60 \\ D-53115 Bonn \\ GERMANY}
\email{kepka@iam.uni-bonn.de}

\author{\large Juan J. L. Vel\'{a}zquez}
\address{University of Bonn, Institute for Applied Mathematics \\ Endenicher Allee 60 \\ D-53115 Bonn \\ GERMANY}
\email{velazquez@iam.uni-bonn.de}

\date{}

\begin{document}
	\maketitle

\begin{abstract}
    In this paper we study how to determine if a linear biochemical network satisfies the detailed balance condition, without knowing the details of all the reactions taking place in the network. To this end, we use the formalism of response functions $R_{ij} (t) $ that measure how the system reacts to the injection of the substance $j$ at time $t=0$, by measuring the concentration of the substance $i \neq j$ for $t >0$. In particular, we obtain a condition involving two reciprocal measurements (i.e.~$R_{ij}(t), \, R_{ji}(t)$) that is necessary, but not sufficient for the detailed balance condition to hold in the network. Moreover, we prove that this necessary condition is also sufficient if a topological condition is satisfied, as well as a stability property that guarantees that the chemical rates are not fine-tuned. 
\end{abstract}

\begin{keywords}
    Detailed balance; biochemical network; response function; cut vertex.
\end{keywords}

\tableofcontents

\section{Introduction}
A common feature of most of the subsystems that are parts of a biological organism is that they operate out of equilibrium. Therefore they must exchange energy with their surroundings in order to function. 

The most common molecule used by cells to transport energy is Adenosine triphosphate (ATP), although other molecules as Guanosine triphosphate (GTP), Nicotinamide Adenine Dinucleotide Phosphate (NADPH) are sometimes also used. 
The energy stored in ATP, or in any other aforementioned molecules, is realised during a chemical reaction via the dephosphorilization of ATP in Adenosine diphosphate (ADP). 

We now mention few examples of biological systems for which it has been conclusively shown that they work in out of equilibrium conditions. A well known case are the kinetic proof-reading mechanisms of the immune system. 
These mechanisms allow the immune system  to obtain a much larger discrimination between antigens, than the one that could be expected from the difference of their affinities with the receptors of the immune cells. This was studied by Hopfield \cite{hopfield1974kinetic} and Ninio \cite{ninio1975kinetic}. 
Similar mechanisms take place in the recognition and error correcting systems acting in processes like mRNA transcription, DNA duplication  and others, see \cite{Boeger} for a review on the topic.

Another example  of biological system for which it has been experimentally shown that the system works out of equilibrium is the actively beating Chlamydomonasflagella (see \cite{battle2016broken}). 
In general all the systems involving molecular motors, like kinesins, which are involved in the transport of proteins inside the cells, or like the actin/myosin motors, which act as propellers for eukaryotic cells \cite{chaffey2003alberts}, \cite{magnasco1994molecular}, or bacteria like Listeria \cite{tilney1989actin}, must necessarily work out of equilibrium in order to transport molecules or to produce mechanical work to yield cell displacement.

As another example of biological system working out of equilibrium we refer to the pump and leak mechanism, that is used by cells to maintain a constant size, despite the huge difference of osmotic pressure between the interior and the exterior of the cell, see \cite{tosteson1960regulation} as well as \cite{keener2009mathematical} and \cite{mori2012mathematical}. 
A related fact is the way in which the family of ABC proteins, which are on the membrane and acquire different configurations, 
allow molecules to traverse the membrane. The concentrations of substance on both the sides of the membrane are out of equilibrium, therefore, 
in order to be able to select the molecules allowed to cross the membrane, the ABC proteins need to use energy in the form of ATP (see \cite{flatt2023abc}).
We remark also that it has been suggested (see \cite{kurbel2011donnan}) that the values of the Voltage potential on the membrane of some red blood cells are created by a passive mechanism, known as Donnan effect, without requiring active sodium/potassium pumps.

As we will explain in detail later, in this paper we are interested in understanding if a biochemical system satisfies the \textit{detailed balance condition} or not, by measuring the evolution of some suitably chosen substances in the system. A biochemical  system satisfies the detailed balance condition if all the reactions are balanced at the stationary state. 
We refer to Section \ref{sec:DB} for a precise definition of the detailed balance property for the type of systems of interest for this paper. 

An essential feature of all the biological systems introduced above, that require the use of energy, is that they function due to the fact that they are open systems, working in out of equilibrium conditions. More precisely, in all the examples above, the system is able to work because there exists a continuous influx of molecules able to store energy (for instance ATP) and a continuous outflux of by-product (for instance ADP).
However, we can approximate the dynamics of these open systems by means of effective systems that are closed and for which the detailed balance fails, see Section \ref{sec:lack of DB} for more details on this approximation. 
This is the reason why we are interested in studying systems that do not satisfy the detailed balance condition.

In this paper we study how to determine if a linear biochemical system satisfies the detailed balance condition. To this end, we use the formalism of the response functions, introduced in \cite{thurley2018modeling}. Some mathematical questions referring to the response function formalism were studied in \cite{franco2023description}. 
The main idea behind this approach is to study the main properties of a linear biochemical system analysing how the system reacts to certain inputs. 
More precisely, we study a biochemical system by analysing the properties of the response functions $R_{ij}(t)$. These response functions describe the evolution in time of the concentration of the substance $i$, after the arrival of a molecule of substance $j$ to the system. 
The advantage of studying a biochemical system via the response functions $R_{ij}(t)$ is that these functions can be, at least in principle, measured, without requiring a detailed knowledge of all the rates of the reactions that take place in the biochemical system. 

 We restrict the attention to the case of linear biochemical systems with a finite number of states, say $L\in\N$. 
This simplifies the analysis due to the fact that we have a clear correspondence between systems of ODEs and graphs, that will allow us to deduce results using graph theory. 
The relation  between graphs and biochemical reaction networks is more involved in the non-linear case, see \cite{dal2023geometry,feinberg2019foundations}. 

We stress that the analysis of biochemical systems, that do not satisfy the detailed balance property, is a very active area in bio-physics, see for instance \cite{battle2016broken,jiang2003entropy,li2019quantifying,martinez2019inferring,zia2007probability}. 
The novelty of this paper is that we approach the problem using the formalism of response functions.
Although, as we explain in Section \ref{sec:stable properties}, the results obtained in this paper can be understood also as results for stochastic processes, making the connection with the approach studied in \cite{martinez2019inferring}.

In this paper we prove that, if a linear biochemical system satisfies the detailed balance condition, then for every states $i, j $ there exists a constant $c>0$ such that
\begin{equation} \label{eq:intro p(DB)}
R_{j i } (t) =   c R_{i j }(t) \text{ for every } t \geq 0. 
\end{equation}
We remark that \eqref{eq:intro p(DB)} provides a necessary condition for the detailed balance property. In particular, consider a linear biochemical system, with a finite number of states, such that there exists a couple of states, say $1 , 2$, and two positive times $t_1 \neq t_2$, such that $R_{12} (t_1)/R_{21 }(t_1) \neq R_{12} (t_2)/R_{21 }(t_2)$. Then the system does not satisfy the detailed balance condition. 

It is then natural to ask the following question. Assume that equality \eqref{eq:intro p(DB)} holds for a particular choice of states $i $ and $j$. Can we infer that the detailed balance condition holds? 
If not, can we infer the detailed balance condition by  performing more measurements? 

We prove that, if the number of states in the system, $L$, is equal to $3$, then the fact that two states in the system satisfy the property \eqref{eq:intro p(DB)} is equivalent to the detailed balance property of the system.
However, this result holds only for $L=3 $ and we prove it to be false for $L>3$, via a counterexample in $L=4$ (see Example \ref{exam:L=4 strongly connected}).

Nevertheless, since we are dealing with biological applications, it is important to understand if the fact that \eqref{eq:intro p(DB)} holds for a biochemical network, for some $i$ and $j$, is a \textit{robust} property. 
Indeed, small mutations can be expected to modify the chemical rates of a system. Therefore, chemical systems, that do not satisfy the detailed balance property, and satisfy \eqref{eq:intro p(DB)} for some $i,j$ only for some \textit{fine-tuned} rates, cannot be expected from the biological point of view, unless one finds an underlying reason for the network to have fine-tuning.

This motivates us to introduce a definition of stability classes of biochemical systems. A stability class is just a set of biochemical systems involving the same set of substances and whose reactions share the network architecture. 
In particular, if a biochemical system in a certain class is such that the substances $i$ and $j$ do not interact with each other (namely an element with state $i$ cannot jump to state $j$ and vice versa), then the same will be true for all the biochemical systems in the class. 
Moreover if at stationary state a biochemical system is such that the reactions between the substances $i$ and $j$ balance, then the same will be true for all the biochemical systems in the class.

These  features are properties that are not easily modified by evolution, because they would require a very substantial mutation or a very large number of mutations. 
In other words, small mutations, leading to changes in the coefficients of a reaction network, will not modify the class of the network.  
Therefore, we will say that a property is stable or robust in a certain class, if small changes in the reactions rates do not destroy the property. 
We refer to Section \ref{sec:class and stability} for the precise definition of stability classes. 

One could impose different concepts robustness and of classes of networks from the ones that we introduce in this paper.
However, we will see that the definitions that we give not only have some biochemical rationale, but they are also flexible enough to allow us to prove rigorous mathematical results.

In particular we prove that the fact that \eqref{eq:intro p(DB)} holds for some pair $i,j $ is an unstable property for many of the stability classes of biochemical systems that do not satisfy the detailed balance condition. 
This means that if we slightly modify the reactions rates of the biochemical network in one of these stability classes, then we can obtain a biochemical network belonging to the same class, for which \eqref{eq:intro p(DB)} does not hold for any $i,j$. 
This means that to have a biochemical system belonging to this class and satisfying \eqref{eq:intro p(DB)} for some pair $i,j$, without satisfying the detailed balance condition, requires a fine-tuning of the parameters in the system.

Hence, from the practical point of view, if we assume that fine-tuning of the chemical rates does not take place and if we know a priori that a biochemical system belongs to a stability class with a suitable architecture and we have a couple of reciprocal measurements $R_{ij},\, R_{ji} $ satisfying \eqref{eq:intro p(DB)}, then we can infer that the detailed balance condition is satisfied.
An example of class of biochemical systems for which \eqref{eq:intro p(DB)} is unstable is the class of systems in which each of the substances interact with all the other substances in the system and there is at least one reaction in which detailed balance fails, i.e.~one reaction is not balanced at the stationary state. See Corollary \ref{cor:unstable PDB complete graph with one detailed balance} in Section~\ref{sec:stable pdb implies cut vertex}. 

We also prove that there exists a class of biochemical systems that do not satisfy the detailed balance condition and that satisfy \eqref{eq:intro p(DB)} for some states $i,j$ in a stable way. 
This class of system is characterized by a topological property of the graph induced by the reactions among the substances in the system. 

The specific property of the architecture of the network that allows us to have \eqref{eq:intro p(DB)} in a stable manner, without satisfying the detailed balance condition, is the existence of a cut vertex in the graph corresponding to the reaction network. A \textit{cut vertex} is a vertex of the graph with the property that removing it, as well as all the edges connected to it, results in a disconnected graph. The consequence of this is that, if the property $\eqref{eq:intro p(DB)}$ holds for one couple of vertices $(i, j)$, then the detailed balance condition is satisfied, unless either the biochemical system has a particular architecture or the reaction rates are fine-tuned.

The results of this paper can also be interpreted in terms of so-called hidden Markov processes. Let us mention that the detailed balance condition is referred to as reversibility in the theory of stochastic processes.
The equations under consideration in this paper describe also the evolution of the probability densities of a set of states of a Markov process. Hidden Markov processes are characterized by the fact that only some of the states are visible. If both states $(i,j)=(1,2) $ are in the observable part of the Markov process, then the condition \eqref{eq:intro p(DB)} indicates an observable condition that allows to determine (for stable systems in the sense explained above) if the detailed balance condition holds. Finally we recall that hidden Markov processes are a particular type of semi-Markov processes, for which the transition probabilities are given by suitable integral operators, see \cite{feller1991introduction}.

The results in this paper have some analogies with the theory of inverse problems, in the sense that we try to determine if a property  (reversibility, i.e.~detailed balance) holds for a system, by performing a set of measurements.
The type of equations considered in this paper can be thought as parabolic equations in a discrete setting. The theory of inverse problems for elliptic and parabolic equations is very well developed. We refer to the book \cite{isakov2006inverse} for examples of problems of determining the coefficients of an elliptic or parabolic equation from a small set of measurements. However, we are not aware of results determining the reversibility of the underlying stochastic process from some measurements. 
It would be natural to consider the type of problems studied in this paper in the continuous setting (i.e.~to determine the reversibility of some stochastic process from some particular measurements). 

In this paper, we are mainly interested in the study of biochemical systems that are closed, i.e.~we assume that the number of substances in the system does not change in time. 
In Section \ref{sec:extended db} we explain how to generalize some of the concepts and the results of the paper for nonconservative systems, i.e.~some substances are degraded or added to the system.
More precisely, we introduce the notion of extended detailed balance and we prove that \eqref{eq:intro p(DB)} is a necessary condition to have that the extended detailed balance condition holds. 

Finally, in Section \ref{sec:non reverse measurements} we illustrate  why reverse measurements, as in \eqref{eq:intro p(DB)}, are the natural types of measurements to study if a system satisfies or not the detailed balance condition. 
Indeed if these two measurements are such that \eqref{eq:intro p(DB)} is false, then we can already deduce that the system does not satisfy the detailed balance condition. 
In Section \ref{sec:non reverse measurements} we show that it is not possible to formulate a necessary condition for the detailed balance in terms of two measurements that are not mutually reversed, as they are in \eqref{eq:intro p(DB)}. 
In Section \ref{sec:non reverse measurements} we also show that, to be able prove that \eqref{eq:intro p(DB)} implies the detailed balance condition without any stability assumption on the rates, one needs to verify \eqref{eq:intro p(DB)} for a number of reciprocal measurements that is of the order of $L/2 $ when $L$ is large. Here, $L$ is the size of the system. 
Hence, when $L $ is large, this approach would not be feasible, at least from the practical point of view. 

We wish to point out that the theory of biochemical systems is a well developed area and a huge amount of work has been done (also for non-linear systems) in order to relate the behaviour of the solutions to chemical networks in terms of topological properties of the graph, see \cite{dal2023geometry,feinberg1972complex, feinberg1972chemical,feinberg2019foundations}. 

For the reader's convenience we state here the main result of this paper in an informal way. A  precise statement is given in Theorem \ref{thm:stable pdb implies articulation}.

\begin{theorem}
    Consider a biochemical network without cut vertices for which property $\eqref{eq:intro p(DB)}$ holds and is stable under small, admissible perturbations of the chemical rates. Then, the detailed balance condition holds.
\end{theorem} 
Here, admissible perturbations are changes that keep zero reaction rates which were originally zero and keep the detailed balance property on  a set of reactions for which the detailed balance property is expected to hold, see Section \ref{sec:class and stability}.

\subsection{Plan of the paper}
In Section \ref{sec:linear biochemical systems} we introduce the definition of linear biochemical system, we include the definition of detailed balance in this context and we introduce some basic results, useful to understand if a system satisfies the detailed balance condition or not in particular situations. 

In Section \ref{sec:p DB} we introduce the definition of pathwise detailed balance, we explain its relations with the detailed balance condition and with equality \eqref{eq:intro p(DB)}. We also provide a graphical interpretation of the pathwise detailed balance condition, that will require us to introduce also the concept of pathwise symmetry. 

In Section \ref{sec:class and stability} we introduce the concept of stability of a property in a suitably defined class of networks. 
In this section we also explain  in detail the biological motivation for the definition of class that we propose and the reasons why the stability property is relevant for biological applications. The stability allows to give a precise mathematical definition of the concept of robustness and fine-tuning of the parameters. Finally, we prove that the pathwise detailed balance condition is stable for a set of classes. 

In Section \ref{sec:stable pdb implies cut vertex} we prove that the existence of a cut vertex is a necessary property for the stability of the pathwise detailed balance condition. 
This implies the instability of the pathwise detailed balance property for all the classes of networks that do not have a cut vertex. 

In Section \ref{sec:Stochastic} we give an interpretation of \eqref{eq:intro p(DB)} in terms of stochastic processes. 
In Section \ref{sec:extended db} we extend the concept of detailed balance to nonconservative systems, i.e. systems without mass conservation. 
In Section \ref{sec:non reverse measurements} we explain why we concentrate on reverse measurements of the form \eqref{eq:intro p(DB)}. 
Finally, we provide some concluding remarks in Section \ref{sec:conclusion}.

\subsection{Notations}
We define $\R_+ := (0,\infty) $, $\R_*=[0, \infty)$ and write 
\begin{equation} \label{def e} 
\textbf{e}_n =(1, \dots, 1 )^T \in \mathbb R^n.
\end{equation}
Vectors will always be column vectors, i.e.~$v \in \mathbb R^n$ is such that $v=(v_1 \dots, v_n)^T $. 
We denote the scalar product between two vectors $v_1, v_2 $ as $\langle v_1 , v_2 \rangle $. We use the notation $ v_1 \otimes v_2  $ to indicate the tensor product between the two vectors $v_1, v_2.$
We use the following notation 
\[
D=\operatorname{diag} (\{ v_i \}_{i=1}^n )
\]
to denote the diagonal matrix $D \in \mathbb R^{n \times n } $ with $D_{ii}= v_i $ and $D_{ij}=0$ for $i\neq j$. Furthermore, we write $\norm{A}$ for the norm of a matrix $A\in \R^{n\times m}$ with respect to any matrix norm. Finally, for a matrix $A\in \R^{n\times m}$ we denote by $A^T\in\R^{m\times n} $ its transpose.

\section{Linear biochemical systems and detailed balance}\label{sec:linear biochemical systems}
In this section we give a precise definition of the biochemical systems that we consider in this paper.
Then we introduce the notion of detailed balance as well as some necessary or sufficient conditions that allow, in some specific contexts, to check if the detailed balance condition holds (see Section \ref{sec:DB}). We also explain why it is relevant to analyse biochemical systems that do not satisfy the detailed balance condition (see Section \ref{sec:lack of DB}). 
Before doing that we recall some notions of graph theory that are used in this paper (see Section \ref{sec:graph}). 

\subsection{Notions of graph theory} \label{sec:graph}
We recall here some notions of graph theory that will be used in this paper, see e.g. \cite{Bondy1976graph}. 
Let $\mathcal G= (V,  E)$ be a directed graph with vertices $V$ and edges $E \subset V \times V $.
Given an edge $e=(e_1, e_2)$ we say that its ends are the vertices $e_1$ and $e_2$. We use the following notation for the reverse edge,  $e^*:=(e_2, e_1)$.
We say that $\mathcal G $ is a \textit{symmetric directed graph} if $e =(e_1, e_2) \in E $ implies that $e^*=(e_2 , e_1 ) \in E$. 
Notice that we can associate to $\mathcal G $ a (undirected) graph, by replacing every directed edge $(e_1, e_2) \in E $ with an undirected edge and removing the reverse edge $(e_2, e_1)$. 
Finally, a graph is the \textit{complete} graph if $E=V^2$.

We recall now the definition of \textit{walk}, \textit{path} and \textit{cycle}. 
A walk $w $ in $\mathcal G$ is a finite non-null sequence $v_0 e_1 v_1 e_2 v_2 \dots e_k v_k $ whose terms are alternatively vertices and edges such that, for $1 \leq i \leq k$, the ends of
$e_1$ are $v_{i-1}$ and $v_i$, i.e.~$e_1=(v_{i-1}, v_i)$.
If the first vertex of $w$ is $v_0$ and the last is $v_k$, we say that $w$ is a walk from $v_0$ to $v_k$, or a $(v_0, v_k)$-walk. 
Moreover, $\ell(w):=k$ is the length of the walk. We denote the set of walks in $\mathcal G$ by $W$. Furthermore, we write $W_{v u}^{(n)} $ for the set of $(v, u)$-walks of length $n$. 
If the edges and the vertices of a walk $w$ are distinct, then $w$ is called a path. 
A walk is closed if it has positive length and its origin and terminus are the
same. A closed walk, without repeated edges, whose origin and internal vertices are distinct is a cycle. 
Let $w $ be a walk and let $v \in V $. 
We say that $v \in w $ if the sequence $w$ contains the vertex $v$. Similarly we say that an edge $e$ is such that $ e \in w $ if $e$ belongs to the sequence $w$.
Moreover, if $e \in w $, we denote with $w \setminus \{  e \} $ the set of the edges in $w$ different from $e.$
Finally we denote with $w_1 \cap w_2 $ the set of the vertices that belong both to $w_1 $ and to $w_2$. 
A $(v_i, v_j)$\textit{-section of a walk} $v_0 e_1 v_1 e_2 v_2 \dots e_k v_k $ is a walk that is a subsequence $v_i e_{j+1} v_{i+1}, \dots e_j v_j$
of consecutive terms of $w$. We will use the notation $w_{(v_i, v_j)} $ to indicate the $(v_i, v_j)$-section of $w$. 
Let $\gamma $ be a walk. Assume that $v_1, v_2 \in \gamma $ then we use the following notation
\begin{equation}\label{distance between wertices on a path}
    \operatorname{dist}_\gamma (v_1,v_2) = \ell(\gamma_{(v_1, v_2)} ) 
\end{equation}
Let $w_1 $ be a $(v_0, v_k)$-walk and let $w_2$ be a $(v_k, v_n)$-walk. 
Then $w:=w_1 \oplus w_2 $ is the walk obtained by concatenating $w_1$ and $w_2$, i.e.~$w = v_0\ldots v_k\ldots v_n$.

Two vertices $u$ and $v$ of $\mathcal G$ are said to be connected if there is a $(u, v)$-path in $\mathcal G$.
If  every pair of vertices in a graph are connected then we say that the graph $\mathcal G $ is a \textit{connected graph}. 
A \textit{tree} is a connected undirected graph with no cycles. A tree is a \textit{spanning tree} of a graph $\mathcal G$ if it includes every vertex of $\mathcal G$ and is a subgraph of $\mathcal G$. 
A \textit{cycle graph} is a graph that consists only of a cycle.

\subsection{Linear biochemical systems} 
We consider a finite chemical system and denote its state space by $\Omega :=\{1, \dots, L\} $, where $L$ is the number of states in the system. Furthermore, we assume the reactions are linear. Hence we deal only with reactions of the form  
\[
(i) \rightarrow (j) 
\]
that take place at rate $A^i_j \geq 0$. 

The dynamics of the concentrations  $n(t)=(n_1(t),\ldots, n_L(t))^T \in \mathbb R_*^L$ is given by the following system of ODEs
\begin{equation} \label{eq:ODE}
\frac{d n }{dt} = A n, \quad n(0)=n_0,
\end{equation}
with $n_0 \in \mathbb R_*^L$. 
In this paper we assume that the matrix $A\in \mathbb{R}^{L \times L }$ is \textit{Markovian}, i.e.~
\begin{equation} \label{structure of A} 
	A_{ii}=-\sum_{k\in \Omega \backslash \{i\}}A_{ki} =-\sum_{k\in \Omega \backslash \{i\}}A^i_{k}  \ \ \text{for all }i\in \Omega \ \ ,\ \ A_{ki}=A^i_k \geq 0\ \text{for }k\neq i.
\end{equation}
The Markovianity property of $A \in \mathbb R^{L\times L}$ implies that $\textbf{e}_{L}^TA=0$, where $\textbf{e}_L$ is given by \eqref{def e}, hence the total number of elements in the system is conserved. Another way of interpreting the solution $n(t)$ to the system of ODEs is the probability distribution of a pure-jump Markov process with rates $A_{ij}$ and initial distribution $n_0$.

The system of ODEs \eqref{eq:ODE} induces a graph structure on the state space $\Omega$. More precisely, we define the set of edges $ E\subset \Omega^2 $ by
\begin{align*}
    E=\setdef{(i,j)\in\Omega^2 }{A^i_j \neq 0}.
\end{align*}
The graph $\mathcal G(A)= (\Omega , E)$ is directed. Notice that this graph includes also self-loops. If we associate to each directed edge $(i,j)$ the weight $A^i_j$, then $\mathcal G(A) $ is a weighted and directed graph. We will refer to the weighted graph (i.e.~including the values of the chemical rates) as biochemical network.

In this paper we assume that the graph $\mathcal G(A) $ is connected, i.e.~the undirected underlying graph obtained by replacing all directed edges of $\mathcal G(A)$ with undirected edges is a connected graph. 

Moreover,  we will assume that the matrix $A$ is \textit{ergodic}, i.e.~we assume that the underlying Markov process is irreducible, in other words, every two states can be connected by a walk in the graph associated to $A$, and hence the system of ODEs \eqref{eq:ODE} has a unique steady state $N=(N_1, \dots, N_L)^T \in \mathbb R_+^L$ with $\| N \|_1 =\sum_{i=1}^L N_i=1 $ such that $AN=0$. Furthermore, $n(t) \rightarrow N $ as $t \rightarrow \infty$ due to the Perron-Frobenius theorem, see e.g. \cite[Chapter VIII]{feller1991introduction}.
For notational reasons it is therefore useful to define the following set 
\[
\mathcal A(L):=\{  A \in \mathbb R^{L\times L }: A  \text{ is Markovian and  ergodic} \}.  
\]
To simplify the notation, we sometimes abbreviate $\mathcal A(L)$ with $  \mathcal A$. We now define the response functions in terms of the matrix $A$. 
Given $i,j \in \Omega $ we define the \textit{response function} $R_{ij}$ as 
\[
R_{ij} (t) := 
\langle e_j , e^{ t A  } e_i \rangle \text{ for } i, j \in  \Omega. 
\]
In other words $R_{ij} $ is the $j$-th component of the solution  $n(t)$ of the ODE \eqref{eq:ODE} with initial condition $n_0 = e_i$. Here, $e_i$ is the $i$-th vector of the canonical basis in $\mathbb R^L$. 
Therefore, as indicated in the introduction, the response function $R_{ij}(t)$ measures the response of the chemical system to a signal, which is just the injection of the substance $i$ at time zero, by measuring the changes in time of the concentration of the substance $j $.
We notice that this formula for the response function is in agreement with the expression of response functions corresponding to systems of ODEs derived in \cite[Lemma 2.9]{franco2023description} by considering the chemical system as a compartment with influxes in $i $ and outfluxes in $j$. 

The response function $R_{ij}$ has an interpretation in the theory of hidden Markov processes. Indeed, these are stochastic processes for which there is an underlying Markov process, but only a small number of states are observable. For instance, suppose that in the Markov process with probability densities solving \eqref{eq:ODE} only the states $1$ and $2$ are observable. Then, the response functions $R_{12},\, R_{21}$ will be observable, but not the other response functions $R_{ij}$ with $(i,j)\neq(1,2), \, (2,1)$. It is then natural to derive a condition involving only $R_{12},\, R_{21}$  in order to determine if the detailed balance condition holds. Hidden Markov processes have been extensively applied in mathematical biology (see for instance \cite{martinez2019inferring}). 

In the paper it will be, in some cases, useful to divide a biochemical systems in compartments and to analyse the concentrations of elements in each of the compartments.
To this end it is useful to introduce the following notation. 
Let $A \in \mathbb R^{ L \times L } $. 
Let $X=\{ \alpha_1, \dots, \alpha_n \} $ be a partition of $\{ 1, \dots, L\}$, i.e.~$\Omega = \cup_{i=1}^n \alpha_i $ with $\alpha_i \cap \alpha_j = \emptyset$ for every $i \neq j $. 
We define the following blocks of the matrix $A.$
Assume that $\beta, \gamma \in X $. 
Then we define the matrix $A_{\beta\gamma} \in \mathbb R^{|\gamma | \times| \beta |}$
as
\begin{equation} \label{A alpha beta}
( A_{\beta \gamma })_{ji} := A_{ji} \quad  \text{ for } \quad i \in  \gamma,\ j \in \beta. 
\end{equation}
The matrix $A_{\beta \gamma}$ encodes the interactions between the compartment $\beta $ and $\gamma$. 
Moreover, we define the matrix $E_\beta $ as 
\begin{equation}\label{A alpha alpha}
(E_\beta )_{ji} := A_{ji} \quad \text{ for } i, j \in \beta  , \text{ with } i \neq j \text{ and } (E_{\beta} )_{ii} := - \sum_{ j \in \beta} A_{ji}. 
\end{equation}
Therefore the matrix $E_\beta $ describes the evolution within the compartment $\beta $ ignoring the interaction with states outside $\beta$.
Finally, given a compartment $\beta $ we define the matrix 
\begin{equation} \label{A loss}
C_\beta = \operatorname{diag}\left( \left\{ \sum_{\alpha \in X \setminus \{\beta\}}\sum_{j\in \alpha} A_{ji} \right\}_{i\in\beta} \right) .
\end{equation}
The matrix $C_\beta $ describes the loss of elements from the compartment $\beta $ to any other compartment.

\subsection{Detailed balance} \label{sec:DB}
In this section we introduce the definition of detailed balance for the matrix $A$ and explain that, the fact that the matrix $A$ satisfies or not the detailed balance condition, is related to the topology of the graph $\mathcal G (A)$ induced by the matrix $A$. 
\begin{definition}[Detailed balance]\label{def:DB} 
We say that a matrix $A \in \mathcal A(L)$
satisfies the detailed balance condition if $ A_{ij}N_j=A_{ji}N_i $ for all $ i,j\in \Omega $, where $N \in \mathbb R_+^{L \times L}$ is the normalized steady state of $A$.
\end{definition}
We now give an equivalent condition of detailed balance in terms of the matrix $ B=S^{-1}AS$, where
\begin{align} \label{eq:S}
    S:=\operatorname{diag}\left( \{\sqrt{N_j}\}_{j\in \Omega} \right).
\end{align}
Here, $N$ is the unique (with $ \| N\|_1=1$) steady state of the ergodic matrix $A$. Let us mention that the matrix $B=S^{-1}AS$ has the same left and a right eigenvector $(\sqrt{N}_i)_{i=1}^L$ with respect to the eigenvalue zero.
\begin{proposition}\label{prop:DB and Symmetry}
 Assume $A \in \mathcal A(L) $ satisfies the detailed balance condition. Then the matrix $B=S^{-1} A S $ is symmetric and such that 
 $\ker (B) = \ker (B^T) = \operatorname{span}(v)$ where $v_i=\sqrt{N_i}$ for $ i \in \Omega $.
 
 Vice versa, assume that $B \in \mathbb R^{L \times L} $ is symmetric and such that 
  $\ker (B) = \ker (B^T) = \operatorname{span}(v)$. 
  Then the matrix $A:= S_v B S_v^{-1}$, with 
  \begin{equation} \label{eq:Sv}
   S_v := \operatorname{diag} (\{ v_i \}_{i\in \Omega}), 
  \end{equation}
  has the unique steady state $N \in \mathbb R_+^L $ with $N_i=(v_i)^2$ for $ i \in \Omega $. Moreover $A$ satisfies the detailed balance condition.
\end{proposition}
\begin{proof}
    The first statement follows by the definition of $B$ and by the fact that $A$ satisfies the detailed balance condition.
    
    Assume instead $B= B^T $ and $\ker(B)=\ker (B^T) = \operatorname{span} (v) $. Then by definition the matrix $A$ is such that $AN=0$ for $N \in \mathbb R^L $ with $N_i = (v_i)^2$.
    Moreover assume that there exists a vector $w$ such that $A w =0$. Then $S^{-1} w $ would be in $\ker(B) =\operatorname{span} (v) $. Therefore, up to a constant, $w = N $.  
    Hence $S=S_v$ and $A$ satisfies the detailed balance condition. 
\end{proof}

Since we are considering linear biochemical systems, as explained in Section \ref{sec:linear biochemical systems}, we can associate to every system a graph $\mathcal G $ on the set of the states $\Omega $. 
This is very useful since from the geometry of the graphs we can, in some cases, infer that the system satisfies the detailed balance property. 
In the case of non-linear system the correspondence between ODE systems and graphs is more involved and the analysis of the detailed balance property in the non-linear systems is therefore more challenging (see \cite{dal2023geometry}). 

\begin{remark} \label{remark:DB for L=2}
Let $L=2 $. 
Consider a matrix $A \in \mathcal A (2)$. 
If $A^1_2 >0$ and $A^2_1 >0 $, then $A$ satisfies the detailed balance condition. 
\end{remark}

We provide now a sufficient condition that guarantees that \eqref{eq:ODE} satisfies the detailed balance condition. 
This result is a particular case of Proposition 14.3.3 in \cite{feinberg2019foundations} and we refer there for the proof.

\begin{proposition}[Sufficient geometrical condition for detailed balance]\label{prop:DB on trees}
 Assume that $A \in \mathcal A$. 
 Assume that the graph $\mathcal G(A)$ is a connected directed symmetric graph.
Let $\mathcal G$ be the corresponding (undirected) graph. If $\mathcal G$ is a tree, then $A$ satisfies the detailed balance condition. 
\end{proposition}

\bigskip 

Consider a matrix $A \in \mathcal A$ that defines a connected and reversible graph $\mathcal G $. 
Consider any spanning tree $\mathcal T$ of $\mathcal G $. 
Then, by Proposition \ref{prop:DB on trees}, we know that the matrix $A_\mathcal T \in \mathbb R^{L \times L}$, corresponding to the tree, satisfies the detailed balance condition. Therefore there exists a vector $\mu \in \mathbb R_+^L $ with $\sum_{i=1}^L \mu_i =1$ such that $B=S^{-1} A_{\mathcal T} S $ is symmetric, where $S$ is given by \eqref{eq:S} with respect to $\mu$. 
Notice that the vector $\mu $ is not necessary the steady state of $A$, unless, as we will see in the next corollary, $A$ satisfies the detailed balance condition. 
In the following text we refer to the vector $\mu $ as the \textit{energy vector associated with the spanning tree} $\mathcal T$. This terminology refers to the fact that $\mu$ can be written as $(e^{-E_1}, \ldots, e^{-E_L})/Z$, where $E_i$ is the free energy of the state $i$ and $Z$ a normalization constant (see \cite{zia2007probability}).

We now present a useful corollary of Proposition \ref{prop:DB on trees}. 

\begin{corollary}\label{cor:DB implies invariant measures of subtree are the same}
Assume that  $A\in \mathcal A$ and assume that the corresponding graph $\mathcal G=\mathcal G(A)$ is a connected directed symmetric graph. 
Let $N \in \mathbb R_+^L $ be the steady state of $A$. 
Then the three following conditions are equivalent:
\begin{enumerate}[(a)]
    \item $A$ satisfies the detailed balance condition.
    \item Consider any spanning tree $\mathcal T$ of $\mathcal G $. Let $\mu $ be the energy vector associated with $\mathcal T$ and consider the matrix $S:=\operatorname{diag}\left( \{\sqrt{\mu_j}\}_{j\in \Omega} \right)$. Then $B=S^{-1} A S $ is symmetric. 
    \item  The energy vector associated with any spanning tree of $\mathcal G $ is equal to the steady state of $A$. 
\end{enumerate}
\end{corollary}
\begin{proof}
We prove that $(a)$ implies $(c)$. 
By definition of detailed balance $(a)$ implies that
$A_{ij} N_j =A_{ji} N_i$ for every $i,j \in \Omega$. 
On the other hand, Proposition \ref{prop:DB on trees} implies that for any spanning tree $\mathcal T$ of $\mathcal G $ we have that 
${(A_{\mathcal T})}_{ij } \mu_j = {(A_{\mathcal T})}_{ji } \mu_i$ for every $i, j \in \Omega$. Notice that for every $(i, j ) \in \mathcal T $, we have that $A_{ij}={(A_\mathcal T)}_{ij}$ and $A_{ji}={(A_\mathcal T)}_{ji}$. Hence for every $i,j \in \Omega $ we have that $ N_j/N_i = \mu_j / \mu_i $, hence due to the normalization of $\mu $ and $N$ it follows that $N=\mu $. This implies $(c)$. 

We now prove that $(c) $ implies $(b)$. 
Therefore assume that $(c) $ holds.
Then consider a spanning tree $\mathcal T$.  For every $(i,j) \in \mathcal T $ we have that $A_{ij} N_j  = {(A_{\mathcal T})}_{ij} N_j ={(A_{\mathcal T})}_{ji} N_i = A_{ij} N_j$. Hence the detailed balance condition holds on the edges of $\mathcal G$ that belong also to $\mathcal T$. Repeating this argument for all the spanning trees of $\mathcal G $ we deduce that $A$ satisfies the detailed balance condition. By Proposition \ref{prop:DB and Symmetry} statement $(b)$ follows. 

Finally, we prove that $(b) $ implies $(a)$. 
From $(b)$ we have that for every spanning tree $\mathcal T $ it holds that $ A_{ji} {\mu_{\mathcal T}}_i  = A_{ij} {\mu_{\mathcal T}} _j$ for every $i, j \in \Omega$, where $\mu_{\mathcal T}$ is the energy vector associated with $\mathcal T $. This implies that $\mu_{ \mathcal T}= N $ and $A$ satisfies detailed balance.
\end{proof}

\subsection{Lack of detailed balance in biochemical systems} \label{sec:lack of DB}
In this section we explain why it is relevant to study biochemical systems for which the detailed balance condition does not hold. 
We now recall that the reversibility in time of the quantum mechanical equations imply that in all the physical processes the detailed balance condition holds (see \cite{boyd1974detailed}).
Therefore, in principle, one might think that the detailed balance condition holds for all biochemical systems.
However, most of the biochemical systems are open systems, where the constant influxes/outfluxes (for instance of ATP/ADP) keep the system out of equilibrium.

If these outfluxes/influxes of molecules are the dominant factors determining the concentrations of some of the molecules in the system, it is possible to approximate the evolution of the concentration of the chemical substances in the systems by a reduced biochemical network. In this reduced network the evolution of the substances whose concentrations are determined by external fluxes can be ignored and the concentrations of the remaining substances are determined via a biochemical system that does not have detailed balance. 

We can illustrate this issue with an example.
Suppose that one of the reactions in the biochemical system is 
\begin{equation} \label{eq:non linear example}
A + B \leftrightarrows C + D. 
\end{equation}
We assume that the reaction takes place at constant temperature and constant pressure. We will denote the chemical reaction rates for the direct and inverse reaction respectively as $K_+$ and $K_-$. 
Then the detailed balance condition for the reaction \eqref{eq:non linear example} yields 
\[
K_- N_A N_B = K_+ N_C N_D 
\]
where $N_A, N_B, N_C, N_D $ are the equilibrium concentrations at given pressure and temperature. 

On the other hand, suppose that the system is open, in particular assume that there is a an influx and outflux respectively of substances $A$ and $C$ in the system. 
This results in having an approximately constant concentration of substances $n_A $, $n_C$ in the system. Where in general we can have $N_A \neq n_A $, $N_C \neq n_C$. 
In particular, also the stationary values would differ. Namely, the stationary values of $B, C$ in the open system are $\overline N_C \neq N_C $, $\overline N_D \neq N_D $. 
Moreover, in general, we would have that 
\begin{equation} \label{almost lack of DB}
K_- n_A \overline N_B \neq  K_+ n_C \overline N_D. 
\end{equation}

Due to the fact that the concentrations $n_A$ and $n_C $ are constant we can replace the reaction \eqref{eq:non linear example} by the effective reaction
\begin{equation} \label{eq:linear example}
B \leftrightarrows  D. 
\end{equation}
The effective reaction rates are $\overline K_- =K_- n_A  $ and $\overline K_+ = K_+ n_C$. 
Then \eqref{almost lack of DB} implies that, in general, we have 
\[
\overline K_- \overline N_B \neq \overline K_+  \overline N_D. 
\]
Hence the effective reaction \eqref{eq:linear example} does not satisfy the detailed balance condition. 

This argument explains why it is reasonable, in certain situations, to consider systems that do not satisfy the detailed balance condition. The lack of detailed balance is linked to the presence of fluxes injecting and removing substances to the system, hence keeping the system out of equilibrium.

\section{Pathwise detailed balance} \label{sec:p DB}
In this section we introduce the definition of pathwise detailed balance and of pathwise symmetry. In Section \ref{sec:ps and pdb} we study the relation between these two properties. 
Moreover we show that a matrix $A$ satisfies the pathwise detailed balance property if and only if \eqref{eq:intro p(DB)} holds. 
In Section \ref{sec:DB and p DB} we also prove that, if a matrix $A$ satisfies the detailed balance condition, then it satisfies the pathwise detailed balance condition. 
Moreover, we show via an example, Example \ref{exam:L=4 strongly connected}, that the pathwise detailed balance condition does not imply the detailed balance condition. 

\subsection{Pathwise detailed balance and pathwise symmetry} \label{sec:ps and pdb}
We start this section by giving the definition of pathwise detailed balance. 
\begin{definition}[Pathwise detailed balance] \label{def:PDB}
We say that a matrix $A \in \mathcal A(L)$ satisfies the pathwise detailed balance condition with respect to two states $i,\, j \in \Omega$, $i\neq j$ such that
\begin{equation} \label{eq:PDB}
\langle e_i , A^n e_j \rangle = \frac{N_i}{N_j} \langle e_j , A^n e_i \rangle \quad \text{ for every } n \in \mathbb N,
\end{equation}
where $N$ is the steady state of the matrix $A$.
\end{definition}
If a matrix $A$ satisfies the pathwise detailed balance condition with respect to the states $(i,j)$, we can assume without loss of generality that $(i,j)=(1,2)$. Furthermore, when we need to specify that a matrix $A$ satisfies the pathwise detailed balance condition with respect to two specific states, say $1,2$, we will say that it satisfies the $(12)$-pathwise detailed balance. 
As we did for the detailed balance condition, we want to relate the pathwise detailed balance condition of $A$ with the symmetry properties of the rescaled matrix $ S^{-1} A S $. To this end we introduce the following definition. 

\begin{definition} [Pathwise symmetry]\label{def:Ps}
We say that a matrix $B \in \R^{L\times L}$ is pathwise symmetric with respect to two vertices $i,\, j \in \Omega $, $i\neq j$ if
\begin{equation} \label{eq:PS}
\langle e_i , B^n e_j \rangle =\langle e_j , B^n e_i \rangle \quad \text{ for every } n \in \mathbb N.
\end{equation}
\end{definition}
As before, we can assume without loss of generality that $(i,j)=(1,2)$. Furthermore, when we need to specify that a matrix $B$ satisfies pathwise symmetry with respect to two specific states $1,2 \in \Omega$ we say that it satisfies the $(12)$-pathwise symmetry. 

In the following we study the equivalence of pathwise detailed balance and pathwise symmetry. In particular, we prove that, if a matrix $A$ satisfies the pathwise detailed balance condition, then the matrix $B= S^{-1} A S $, with $S$ defined by \eqref{eq:S} is pathwise symmetric. Furthermore, we provide additional conditions on a pathwise symmetric matrix $B$ that guarantee that there exists a matrix $A$ satisfying the pathwise detailed balance condition and $B=S^{-1} A S $. 
\begin{proposition} \label{prop:pdb and ps}
    Assume that $A \in \mathcal A(L) $ satisfies the $(12)$-pathwise detailed balance condition. Then the matrix $B:= S^{-1} A S $, with $S$ defined as in \eqref{eq:S} with respect to $N$, the steady state of $A$, is such that $\ker (B)=\ker(B^T )= \operatorname{span}(v)$ with $v_i = \sqrt{N_i} $ for every $i \in \{ 1, \dots, L \} $. Moreover $B$ is $(12)$-pathwise symmetric. 
    
    Vice-versa assume that $B \in \R^{L\times L}$ is ergodic and is such that $\ker (B)=\ker(B^T )= \operatorname{span}(v)$ and assume $B$ is $(12)$-pathwise symmetric.
    Then $A:= S_v B S_v^{-1} $, where $S_v$ is given by \eqref{eq:Sv}, is such that $A \in \mathcal A $ and $A$ satisfies the $(12)$-pathwise detailed balance condition. 
\end{proposition}
\begin{proof}
    We firstly assume that $A \in \mathcal A $ satisfies the pathwise detailed balance condition. 
    Then by definition of $B$ we have that $AN =0$ implies $B S^{-1 } N =0$ as well as $(S^{-1} N)^T B =0 $. Hence $S^{-1} N = v \in \ker (B) \cap \ker (B^T)$. 
    Now assume that $w \neq v $ is such that $B w =0$. We obtain that $A Sw =0$ which implies that $Sw = c N $ for some $c>0$. Hence $ w \in \operatorname{span} (v)$. Moreover, we can argue similarly when $B^T w =0$. Thus, we deduce that $\ker(B)=\ker(B^T) = \operatorname{span}(v)$. 
    Finally $B$ is pathwise symmetric, indeed for every $n \in \mathbb N$
    \begin{align*}
        \langle e_2 ,B^n e_1 \rangle= \langle e_1 , (S^{-1} A S )^n e_2 \rangle=
        \frac{\sqrt{N_2}}{\sqrt{N_1} } \langle e_1 , A^n  e_2 \rangle  =\frac{\sqrt{N_1}}{\sqrt{N_2} } \langle e_2 , A^n e_1 \rangle= \langle e_2 , (S^{-1} A S )^n e_1 \rangle =\langle e_2 ,B^n e_1 \rangle. 
   \end{align*} 
   Assume now that  $B$ is such that $\ker (B)=\ker(B^T )= \operatorname{span}(v)$ and is pathwise symmetric.
   Then, by the definition of $A$ and the properties of $B$ we deduce that $ \ker(A) = \operatorname{span} (S_v v)$.
   Moreover, since $v^T B=0$ we have that $e^T A= e^T S_v B S_v^{-1} = v^T B S_v^{-1} =0 $. Hence $A\in \mathcal A $ and the steady state of $A$ is the vector $N$ such that $N_i = v_i^2$ for every $i = 1 \dots L $. 
   Then the path symmetry of $B$ implies that for every $n \in \mathbb N$ we have 
   \begin{align*}
        \frac{v_2}{v_1} \langle e_1 , A^n  e_2 \rangle &= \langle e_1 , (S_v^{-1} A S_v )^n e_2 \rangle = \langle e_1 ,B^n e_2 \rangle=\langle e_2 ,B^n e_1 \rangle= \langle e_2 , (S_v^{-1} A S_v )^n e_1 \rangle \\
        & =\frac{v_1}{v_2} \langle e_2 , A^n e_1 \rangle. 
   \end{align*} 
    Hence $A$ satisfies the pathwise detailed balance condition.  
\end{proof}

We now give an interpretation of the detailed balance condition in terms of graphs. 
The following lemma will allow us to interpret the pathwise detailed balance property in terms of the rates of the walks of the graph $\mathcal G =\mathcal G(A)$ induced by $A$. In particular, we can associate to any walk $w$ a number via the following map $\textbf{a}_A: W \rightarrow \mathbb R$ such that 
\begin{equation}\label{rate of walk}
\textbf{a}_A(w)= \prod_{ e=(e_1, e_2) \in w}  A^{e_1}_{ e_2} . 
\end{equation}
Let us mention that also self-loops are included in the above product.

\begin{lemma} \label{lem:PDB graph}
Consider $A \in \mathcal{A}(L)$ and $\mathcal G $ be induced by $A$. Let $1,2\in \Omega $.
Then for every $n \geq 1 $ we have that
\[
\langle e_1, A^n e_2  \rangle =  \sum_{w \in W_{12}^{(n)}}  \textbf{a}_A (w)
\]
where $\textbf{a}_A$ is given by \eqref{rate of walk}. 
\end{lemma}
\begin{proof}
    The statement follows from the definition of $A^n$. 
\end{proof}

\begin{remark} \label{remark:diff A pdb}
When the matrix $A$ is Markovian end ergodic, Lemma \ref{lem:PDB graph} implies that $A$ satisfies the $(12)$-pathwise detailed balance condition if and only if, for every $n \geq 1 $ we have that
\begin{equation}\label{differential}
\Delta_n(A) :=  N_1 \langle  e_2, A^n e_1 \rangle  - N_2 \langle e_1, A^n e_2 \rangle =  N_1 \sum_{w \in W_{12}^{(n)}} \textbf{a}_A (w) - N_2\sum_{w \in W^{(n)}_{21} } \textbf{a}_A(w) =0. 
\end{equation}
where $\textbf{a}_A$ is given by \eqref{rate of walk}. 
\end{remark}

\begin{remark}
As a consequence of Remark \ref{remark:diff A pdb} and of Proposition \ref{prop:pdb and ps} we deduce that $A$ satisfies the pathwise detailed balance property if the graph corresponding to $B= S^{-1} A S $ is such that the sum of the weights of the $(1,2)$-walks of length $n$ is equal to the sum of the weights of the $(2,1)$-walks of length $n$. 

Indeed Proposition \ref{prop:pdb and ps} guarantees that if $A $ satisfies the $(12)$-pathwise detailed balance condition, then $B$ is $(12)$-pathwise symmetric.  
Consider the weighted graph induced by the coefficients of $B$, i.e.~$\mathcal G (B)$.
Then \eqref{eq:PS} holds.     
By Lemma \ref{lem:PDB graph} we have that 
the condition \eqref{eq:PS} is equivalent to 
\begin{equation}\label{pdb and paths}  
 \sum_{w  \in W_{12}^{(n)}} \textbf{a}_B(w) = \sum_{w  \in W_{21}^{(n)}} \textbf{a}_B(w), \quad \forall n \geq 1,
\end{equation}
where $\textbf{a}_B(w)$ is the function defined in \eqref{rate of walk} with respect to the matrix $B$. 
\end{remark}

As suggested by the graphical interpretation of the pathwise detailed balance property, the following lemma states that, to prove that the pathwise detailed balance condition holds for a matrix $A$, it is enough to check the property \eqref{eq:PDB} for $n \leq L-1$. 
\begin{lemma} \label{lem:L-1 powers to prove PDB}
Let  $A \in \mathcal A(L)$. 
If \eqref{eq:PDB} holds for every $n \leq L-1$, then \eqref{eq:PDB} holds for every $n \in \mathbb N$. 
\end{lemma}
\begin{proof}
This follows by Cayley-Hamilton Theorem (see \cite{hefferon2006linear}), which guarantees that for every $n \geq L $ we have $A^n \in  \operatorname{span}(A^i)_{i=1}^{L-1} $. 
\end{proof}

In the following proposition we state that the pathwise detailed balance condition is equivalent to \eqref{eq:intro p(DB)}. 
\begin{proposition} \label{prop:R12=R21 implies PSEUDO DB}
    Assume that $A\in \mathcal A(L)$ with steady state $N \in \mathbb R_+^L$.
  If the matrix $A$ satisfies the $(12)$-pathwise detailed balance condition, then 
    \begin{equation}\label{eq:R12=cR21}
    R_{12 }(t)=\frac{N_2}{N_1}  R_{21}(t) \text{ for all } t \geq 0.
    \end{equation} 
    Vice-versa, if there exists a $c>0$ such that
     \begin{equation}\label{eq:R12=cR21 inv measure}
    R_{12 }(t)=c R_{21}(t) \text{ for all } t \geq 0, 
    \end{equation}
    then $c=N_2/N_1 $ and $A$ satisfies the $(12)$-pathwise detailed balance condition. 
\end{proposition}
\begin{proof}
Assume that the matrix $A$ satisfies the pathwise detailed balance condition. 
Then \eqref{eq:PDB} holds for any $n \in \mathbb N$, hence for every $t>0 $ we have that $\langle e_2 , e^{ t A} e_1 \rangle = \frac{N_2}{N_1} \langle e_1 , e^{ t A} e_2 \rangle $ and \eqref{eq:R12=cR21 inv measure} follows.  

Vice-versa, assume that \eqref{eq:R12=cR21} holds.
Then $\langle e_2 , e^{ t A} e_1 \rangle = c \langle e_1 , e^{ t A} e_2 \rangle $. 
Since the matrix $A$ is ergodic and the steady state is $N$ we deduce that $e^{t A} e_1 \rightarrow N $ and $e^{tA} e_2 \rightarrow N $ as $t \rightarrow \infty. $
Hence $\langle e_2, e^{t A} e_1\rangle  \rightarrow N_2 $ and $\langle e_1 , e^{tA} e_2 \rangle \rightarrow N_1 $ as $t \rightarrow \infty. $ Hence $c= N_2/N_1$. 
Using the Taylor expansion of $e^{tA}$ and identifying powers of $t$ we deduce that for every $n \geq 0 $
\[
\langle e_2, A^n e_1 \rangle = \frac{N_2}{N_1} \langle  e_1, A^n e_2  \rangle.
\]
\end{proof}

\subsection{Pathwise detailed balance and detailed balance} \label{sec:DB and p DB}
In this section we study the relation between the detailed balance condition and the pathwise detailed balance condition. 
First of all we prove that the detailed balance condition implies the pathwise detailed balance condition, see Proposition \ref{lem:DB implies pDB}. 
Then we prove that, if the graph $\mathcal G(A)$ is a cycle graph, then  the matrix $A$ satisfies the detailed balance condition if and only if $A$ satisfies the pathwise detailed balance condition. In particular this implies, see Lemma \ref{cor:L=3}, that, when $L=3$ the pathwise detailed balance condition is equivalent to the detailed balance condition. 

\begin{proposition} \label{lem:DB implies pDB}
        Assume that the matrix $A\in \mathcal A$ satisfies the detailed balance condition.
        Then $A$ satisfies the pathwise detailed balance condition with respect to any states $i,j \in \Omega$, $i\neq j$. 
\end{proposition}
\begin{proof}
   The matrix $A$ satisfies the detailed balance condition, hence the matrix $B=S^{-1} A S $ is symmetric by Proposition \ref{prop:DB and Symmetry}. Here, the matrix $S$ is given in \eqref{eq:S} via the steady state $N$.
   Therefore for every $i,j \in \Omega$ we have that  
   \begin{align*}
\frac{\sqrt{N_j}}{\sqrt{N_i} } \langle e_i , A^n  e_j \rangle &= \langle e_i , (S^{-1} A S )^n e_j \rangle = \langle e_i ,B^n e_j \rangle=\langle e_j ,B^n e_i \rangle= \langle e_j , (S^{-1} A S )^n e_i \rangle \\
& =\frac{\sqrt{N_i}}{\sqrt{N_j} } \langle e_j , A^n e_i \rangle,
   \end{align*} 
   for every $n \in \mathbb N$. 
\end{proof}

As a second step we prove that if $L \geq 3 $, then the detailed balance property is equivalent to the pathwise detailed balance property, if the graph induced by $A$ is a cycle graph. 
\begin{lemma}\label{lem:p-DB iff DB on cycles}
 Let $ L \geq  3 $. Assume that the graph $\mathcal G=\mathcal G(A) $ induced by the matrix $A $ is a directed symmetric cycle graph.
 Then the matrix $A$ satisfies the pathwise detailed balance condition if and only if it satisfies the detailed balance condition. 
\end{lemma}
\begin{proof}
Consider the spanning tree defined by $\mathcal T=(V, E \setminus \{ (1,2) , (2,1) \}$ and consider the associated energy vector $\overline N $. 
Define the matrix $S:=\operatorname{diag} \left( \{ \sqrt{\overline{N}_j } \}_{j\in V } \right) $  and the matrix $B:=S^{-1} A S $. 
Due to the fact that the detailed balance condition holds on $\mathcal T $ we have that $B_{ij} = B_{ji } $ for every $(i,j) \neq (1,2) .$
The pathwise detailed balance condition implies that 
\begin{align*} 
    \langle e_1, B^{L-1} e_2 \rangle &=\langle S^{-1} e_1, A^{L-1} Se_2 \rangle = \frac{\sqrt{\overline N_2}}{\sqrt{\overline N_1 }} \langle e_1, A^{L-1} e_2 \rangle  = \frac{\sqrt{\overline N_2}}{\sqrt{\overline N_1 }} \frac{N_1}{N_2} \langle e_2, A^{L-1} e_1 \rangle 
    \\
    &= \frac{\overline N_2}{\overline N_1 }\frac{N_1}{N_2} \langle e_2, B^{L-1} e_1 \rangle 
\end{align*}
as well as $\langle e_1, B e_2 \rangle = \frac{\overline N_2}{\overline N_1 }\frac{N_1}{N_2} \langle e_2, B e_1 \rangle. $

Now notice that 
\[
\langle e_2, B^{L-1} e_1 \rangle = B_{j_1 1} B_{j_2 j_1} \dots B_{2, j_{L-1}} =  \langle e_1, B^{L-1} e_2 \rangle =  \frac{\overline N_2}{\overline N_1 }\frac{N_1}{N_2} \langle e_2, B^{L-1} e_1 \rangle.
\]
This implies that $1= \frac{\overline N_2}{\overline N_1 }\frac{N_1}{N_2}$, hence $B_{12}=B_{21}$. This implies the detailed balance condition for $A$.
\end{proof}

Finally, if $L =3 $, then the detailed balance property is equivalent to the pathwise detailed balance property. 
\begin{corollary}\label{cor:L=3}
    Assume that $A \in \mathcal A (3) $. Then, the matrix $A$ satisfies the detailed balance condition if and only if it satisfies the pathwise detailed balance condition. 
\end{corollary}
\begin{proof} 
    For $L=3$ a connected graph is either a cycle or a tree. In the first case the claim follows from Lemma \ref{lem:p-DB iff DB on cycles}. In the second case detailed balance is always satisfied due to Proposition \ref{prop:DB on trees}.
\end{proof}

If $L>3$ then the fact that the matrix $A$ satisfies the pathwise detailed balance condition does not imply that $A$ satisfies the detailed balance condition, as we show in the following example. 
\begin{example} \label{exam:L=4 strongly connected}
Consider the matrix 
    \[
    A= \left( \begin{matrix}
        & -3 & 1  & 1  & 1 \\
        & 1  & -3 & 1  & 1 \\
        & 1  & 1  & -3 & 2  \\
        & 1  & 1  & 1  & -4  \\
    \end{matrix} \right). 
    \]
    See Figure \ref{fig:example nonDB but pDB} to visualize the graph induced by the matrix $A$. 
This matrix has steady state $N=\frac{1}{5}(5/4, 5/4, 3/2,1)$. However, its spanning tree with edges $\{(1,3),(1,2),(2,4) \}$ corresponding to the matrix
\[ 
        A_{\mathcal{T}}= \left( \begin{matrix}
        & -2 & 1  & 1  & 0 \\
        & 1  & -2 & 0  & 1 \\
        & 1  & 0  & -1 & 0  \\
        & 0  & 1  & 0  & -1  \\
    \end{matrix} \right).  \]
and has energy vector $\frac{1}{4}(1,1,1,1)$. By Corollary \ref{cor:DB implies invariant measures of subtree are the same} $A$ does not have the detailed balance property. 
However,  $\langle e_2 , A^n e_1\rangle =\frac{N_1}{N_2}  \langle e_1 , A^n e_2 \rangle = \langle e_1 , A^n e_2 \rangle $ for $n=1,2,3 $.
    \begin{figure}[ht]
	\centering
	\includegraphics[width=0.2\linewidth]{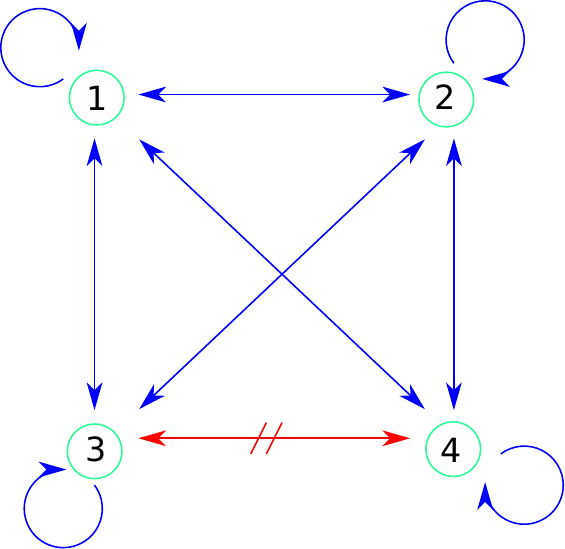}
	\caption{Graph corresponding to the matrix $A$. In red and marked with the symbol $//$ we have the only reaction without detailed balance of the network. 
 }
	\label{fig:example nonDB but pDB}
\end{figure}
\end{example}
This example shows that for $L>3$ the pathwise detailed balance condition does not imply the detailed balance condition. Nevertheless, as we stated in the introduction pathwise detailed balance will be a sufficient condition for detailed balance if we assume a topological property of the graph as well as a suitable defined stability concept on a class of graphs. This is the reason to introduce the definition of stability classes of graphs in the next section.

\section{Stability classes of chemical networks}
\label{sec:class and stability}
In this section we introduce the definition of classes of networks.
This definition will allow us to speak about the stability of certain properties (in particular the pathwise detailed balance condition) of a specific representative of the class. 

\subsection{Stability classes} \label{sec:classes of graphs}
In this section we give the definition of classes of networks. 
This definition has a biological rationale. Specifically, the network classes are defined in such a way that a significant change in the architecture of the network is needed to transform one class into another. 
The class of a network could change in time due to evolution, but this would require a significantly large number of mutations or a relevant mutation.

We now precise what we mean exactly with significant changes in the network, that would require important mutations to take place. 
In the case of linear networks, that we consider in this paper, we assume that significant changes are: 
\begin{enumerate}
    \item the addition of a new substance, that interacts in a non-trivial manner with some of the substances previously available in the network; 
    \item the removal of one of the substances of the network, including all the edges connecting this substance with any of the substances of the network; 
    \item the creation of a new edge connecting two of the previously existing substances in the network; 
    \item the removal of one edge connecting two of the substances present in the network; 
    \item the transformation of some of the edge connections for which the detailed balance condition holds to a new connection where the detailed balance condition does not hold.    
\end{enumerate}

Therefore we  say that the \textit{architecture} of a network is characterized by the set of substances, the set of edges connecting them, and the set of connections for which detailed balance holds. 
Notice that, any of the modifications of the network indicated in the points (a)-(e), implies a modification of the architecture of the network. We assume that important modifications as the ones in (a)-(e) require rare significant mutations or a very large number of small mutations. 

We stress that the definition of architecture of the network defines a topological space on the set of the chemical rates of the network. There are other possible topological spaces on networks that allow more general modifications as the one we introduce in \eqref{eq:alternative ball in a class} below. For instance, it might allow the formation of new weak reactions between substances that were not interacting before. Some of these alternative topologies can be justified by the way in which evolution is expected to produce chemical pathways \cite{noda2018metabolite}.

We define the classes of networks as a set of chemical networks sharing the same architecture. More precisely we have the following definition of class of matrices. 
\begin{definition}[Class of matrices]\label{def:class}
Assume $L \geq 3 $. 
Let 
\begin{equation} \label{def:S} 
\mathcal S:= \{ 1, \dots L \}^2=\Omega^2. 
\end{equation}
Assume that $E_B \subset \mathcal S$, $E_{n-B} \subset \mathcal S $ are such that $E_B \cap E_{n-B} = \emptyset. $
Let $E_N:= \mathcal S\setminus (E_B \cup E_{n-B})$
We define the \textit{class of matrices} $\mathcal A_{E_N, E_B}$ as 
\begin{equation}\label{eq:def class}
    \mathcal A_{E_N, E_B} : = \left\{\ \parbox[c]{.70\linewidth}{ $A \in \mathcal A(L) $ with normalized  steady state $ N \in \mathbb R_+^L  :  \  A_{\alpha_2 \alpha_1}=A_{\alpha_1, \alpha_2} =0, \\ \forall \alpha \in E_N$  and $  A_{\alpha_2 \alpha_1} N_{\alpha_1} = A_{\alpha_1 \alpha_2} N_{\alpha_2}, \forall \alpha \in E_B$  }\ \right\}. 
\end{equation}

\end{definition}
Since to every matrix $A \in \mathcal A_{E_N, E_B} $ we can associate a graph $\mathcal G(A)$, a class of matrices corresponds to a class of networks. 
According to Definition \ref{def:class}, a class of networks is characterized by the set of nodes $\Omega=\{ 1, \dots , L \} $, a set of edges for which the detailed balance condition holds, that is $E_B $, and a subset of edges for which the detailed balance condition is not imposed, that is $E_{n-B}$. 
The network is then uniquely characterized by means of the chemical coefficients $K_\alpha>0$ for the reaction $\alpha \in E_B \cup E_{n-B} $. We have that $K_\alpha =0$ if $\alpha \in E_N$. 

\begin{example}
    Let $A_1, A_2  \in \mathcal A $ satisfy the detailed balance condition and assume that $\mathcal G(A_1) $ and $\mathcal G(A_2) $ are complete graphs.
    Then $A_1, A_2 \in \mathcal A_{\emptyset, \mathcal S} $. 
\end{example}

Notice that in each network class we can define a distance between two networks characterized by chemical coefficients $ \{ K_\alpha \}_{\alpha \in E_B \cup E_{n-B}}$ and $ \{ \overline{K}_\alpha \}_{\alpha \in E_B \cup E_{n-B}}$, using for instance the euclidean distance between these vectors. It is natural to assume that most of the mutations are small, i.e.~they preserve the architecture of the network and just modify the coefficients $\{ K_\alpha \} $ by a small amount. 
On the other hand, some rare mutations or a combinations of several small mutations can lead to a change in the architecture of the network.

The significant modifications in (a)-(e) have been chosen with the idea that such changes would require important mutations (or long sequences of small mutations). 
We assume that addition or elimination of substances or edges in the network requires substantial changes. 
The same applies also to modifications of reactions with detailed balance that result in reactions without detailed balance and vice versa. 

This might be expected because the reactions in which the detailed balance condition holds are usually associated to interactions with some additional molecules, typically ATP, see Section \ref{sec:lack of DB}. The introduction of this additional molecule in one reaction can be expected to be a significant change. This is the reason why we include the set of reactions for which  detailed balance holds or fails as part of our definition of classes of graphs.

\subsection{Stability/instability of a property in a stability class} \label{sec:stable properties}
In this section we introduce the concept of stability of a certain property, say $(P)$, in a class of graphs. 
As explained in Section \ref{sec:classes of graphs}, our definition of network class allows to define a concept of distance between two networks in the same class. 
It turns out that it allows also to define suitable concepts of stability and instability that will be relevant for us later. 

The two reciprocal measurements $R_{12}$ and $R_{21}$ are not sufficient in general to determine if a system satisfies the detailed balance condition or not. 
In Section \ref{sec:non reverse measurements} we show that measurements that are not reciprocal are less efficient to determine the detailed balance condition compared to the reciprocal measurements $R_{12} $ and $R_{21}$. 
Moreover, we show that in order to be able to deduce detailed balance from pathwise detailed balance we need a number of measurements of the order of $L$, recalling that $L$ is the size of the state space $\Omega$. This means that, from the practical point of view it is impossible to prove detailed balance from pathwise detailed balance for large networks. This is the reason why, instead, we focus on analyzing the stability of pathwise detailed balance.

Consider a class $\mathcal A_{E_N , E_B } $ for some choice of $E_N , E_B \subset \mathcal S $ where $\mathcal S$ is as in \eqref{def:S}. Given any $K \in \mathcal A_{E_N, E_B}$ we can define a ball around $K$ in the class of matrices $\mathcal A_{E_N, E_B}$ by
\begin{equation}\label{eq:ball in a class}
    B_\delta (K):= \{ A \in \mathcal A_{E_N , E_B} : \| K - A \| < \delta \} 
\end{equation}
for some $0 < \delta < \infty $. This induces a topology that we denote by $\mathcal T(E_N , E_B)$.

Alternatively, we can also define the following topology, abbreviated by $\mathcal{T}_w(E_B)$, induced by the following open balls around matrices $K\in\mathcal A_{E_N, E_B} $
\begin{equation}
    B_\delta (K):= \{ A \in \mathcal{A}_{E_B^*}: \| K - A \| < \delta, \quad E_B\subset E_B^* \},
\end{equation}
where we define
\begin{align}\label{eq:alternative ball in a class}
    \mathcal{A}_{E_B}:= \{ A \in \mathcal A(L) \text{ with normalized  steady state } N \in \mathbb R_+^L :  A_{\alpha_2 \alpha_1} N_{\alpha_1} = A_{\alpha_1 \alpha_2} N_{\alpha_2}, \forall \alpha \in E_B  \}.
\end{align}
The balls $B_\delta (K) $ in the topology $\mathcal {T}_w(E_B)$ are larger than the ones in the topology $\mathcal T (E_N, E_B)$. Indeed, in the definition of $B_\delta (K) $ for the topology $\mathcal{T}_w (E_B) $ we take into account also of the possible formation of a reaction between substances that, according to the rates in $K$, are not interacting. 
Therefore even if $\mathcal G(K)$ does not have a certain edge, a small perturbation of $K$ in the topology $\mathcal{T}_w(E_B)$ might create this edge with a small reaction rate. 
Nevertheless, all edges in which detailed balance holds are required to satisfy detailed balance. This weaker topology is suitable for biochemical networks in which evolution can produce the formation of weak reactions between substances of the network that were not interacting before, as for instance in \cite{noda2018metabolite}.

\begin{definition} \label{def:stability}
    Assume $L>1$ and let $E_N, E_B \subset \mathcal S $ be as in Definition \ref{def:class}.
    Assume that $A \in \mathcal A_{E_N ,E_B } \subset \mathbb R^{L \times L}$ satisfies the  property $(P)$. Then the property $(P) $ of $A$ is called stable in the class $\mathcal A_{E_N, E_B} $ if there exists an open set $\mathcal O$ of the topology $\mathcal{T}$ containing $A$ and such that $\forall \bar A \in \mathcal O$ property $(P)$ holds. Otherwise, the property $(P)$ is unstable.
\end{definition}
In this paper the topology mentioned in Definition \ref{def:stability} is either $\mathcal{T}(E_N,E_B)$ or $\mathcal{T}(E_B)$. But it would be possible to consider also other topologies. 

Notice that the concepts of stability and instability introduced above is rather natural from the point of view of biological applications. 
Indeed, as explained in Section \ref{sec:classes of graphs}, small mutations can be expected to modify the chemical rates within one given network class, without modifying the class, i.e.~the network architecture. 

Since the specific values of the coefficients characterizing a network are determined by evolution, we cannot expect to have an unstable property $(P)$ for a network, because small mutations would modify the coefficients and would modify the properties of the network. 
In other words, if a network has a certain unstable property, it would require \textit{fine-tuning} of the parameters. 
As a consequence, in the following, we will assume that the only networks that appear in biological applications are the ones which are stable, in the sense of Definition \ref{def:stability} with respect to a topology $\mathcal T$. The specific topology would depend on the way in which evolution generates different networks.

Notice that, in some particular biochemical systems, some specific properties of the proteins involved in the reaction networks or the DNA template yielding these proteins could provide additional restrictions to the values of the chemical coefficients $\{ K_\alpha \} $. 
In such cases, additional constraints should be added to the class of chemical networks, but in principle an approach similar to the one used in this paper should apply to that case. 

We conclude this section by  proving that the set of matrices that do not satisfy the detailed balance condition is an open subset of $\mathcal A$.
In other words we prove that  the lack of detailed balance is stable under small changes of the coefficients $K_{\alpha}$ that characterize the graph, induced by mutations.

\begin{proposition} \label{prop:non-B open set}
Let $L>1 $.
The set 
\[ 
\mathcal A_{n-B}:= \{ A \in \mathcal A: A \text{ does not satisfy the detailed balance condition} \} 
\]
is an open subset of $\mathcal A$. 
\end{proposition}
\begin{proof}
    We prove that the set of Markovian ergodic matrices satisfying the detailed balance condition is closed. To this end, consider a converging sequence of matrices $A_n \in \mathcal A$,  such that $A_n\to A$, with $A \in \mathcal A$, as $n \rightarrow \infty$ and such that $A_n$ satisfy the detailed balance condition for each $n\in \N$. We want to prove that also the limit $A$ satisfies the detailed balance condition. 
    Since the eigenvalue $\lambda=0$ corresponding to the steady state $N$ of $A$ is simple we have $\mu_n\to N$ as $n \rightarrow \infty $, where $\mu_n$ is the steady state of $A_n$. Passing to the limit in the equality
    \begin{align*}
        {(A_n)}_{ij} {(\mu_n)}_j = {(A_n)}_{ji} {(\mu_n)}_i
    \end{align*}
    for all $i,\, j\in \Omega$ yields the detailed balance condition for $A$.
    Hence the set of matrices satisfying the detailed balance condition is closed. In particular, $\mathcal A_{n-B} $ is open. 
\end{proof}

\subsection{Classes of graphs for which the pathwise detailed balance condition is stable}
In this section we introduce a special class of matrices that do not necessarily satisfy the detailed balance condition, but that satisfy the pathwise detailed balance condition in a stable manner.
As we will see in this section, this class of graphs is characterized by a topological feature of the corresponding network, namely the existence of a cut vertex. 

We start this section by recalling the definition of cut vertex.
\begin{definition}[Cut vertex]\label{def:articulation}
Let $\mathcal G=(V,E) $ be a connected graph. 
A vertex $X \in V$ is a cut vertex if $E$ can be partitioned into two non empty subsets $E_1$ and $E_2$ such that $\mathcal G [E_1]$ and $\mathcal G [E _2]$ have just the vertex $X$ in
common. 
\end{definition}

We now introduce a class of matrices that correspond to a graph that has a cut vertex that partitions $E$ in such a way that the edges that do not satisfy the detailed balance condition belong all to the same part of the graph. 
More precisely, the definition is the following. 
\begin{definition}[A family of stability classes with cut vertices]\label{def:Stability class with cut vertex}
    Let $L > 3 $. Let $E_N, E_B \subset \mathcal S$ be as in Definition \ref{def:class}. We say that $\mathcal A_{E_N, E_B} $ is a stability class with cut vertices if $E_N ,\, E_B $ are such that every $\mathcal G(A)$, with $A\in \mathcal A_{E_N , E_B}$, has a cut vertex $X$ such that either $  E_1 \subset E_B$ or $E_2 \subset E_B$ where $E_1, E_2 $ are the two sets in Definition \ref{def:articulation}. We denote by $\mathcal{C}_{E_N , E_B}$ this class of matrices.
\end{definition}

\begin{example}
    \begin{figure}[ht]
	\centering
	\includegraphics[width=0.2\linewidth]{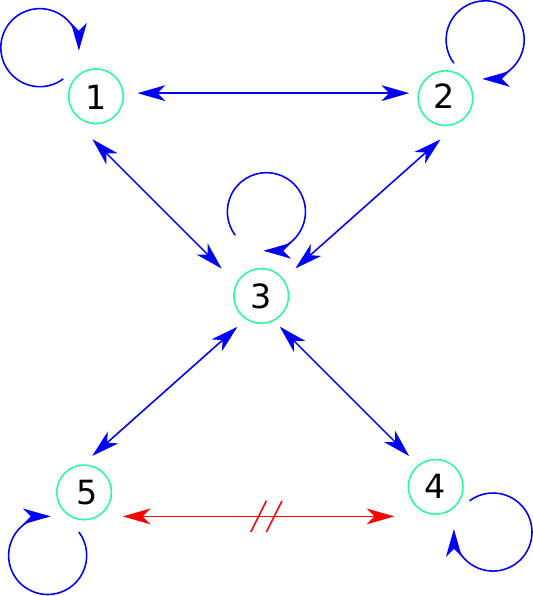}
	\caption{Graph corresponding to the matrix $A$. In red and marked with the symbol $//$ we have the only reaction without detailed balance of the network.}\label{fig:stable class}
\end{figure}
Let $L=5$. Assume that $E_N=\{ (1,4), (2,5) \} $ and assume that $E_{n-B}=\{ (4,5) \} $. See Figure \ref{fig:stable class} for a visualization of the network. Then $\mathcal A_{E_N, E_B } $ is a stability class with the cut vertex $X=3$.
Hence  $\mathcal A_{E_N, E_B } =  \mathcal C_{E_N, E_B } $. 
\end{example}

\bigskip 

\begin{theorem} \label{prop:stable DB}
    Let $L > 3$. Let $E_N , E_B  \subset \mathcal S$ be as in Definition \ref{def:class} and such that the class $\mathcal A_{E_N, E_B} = \mathcal C_{E_N, E_B} $ is a stability class with cut vertices. Every matrix $A \in \mathcal C_{E_N, E_B} $ satisfies the $(12)$-pathwise detailed balance condition with $1,2 \in \mathcal G[E_1]$.
\end{theorem}
\begin{proof}
Let $A \in \mathcal A_{E_N, E_B} $ and let $N$ be the associated (normalized) steady state. 
Define the matrix as $B:= S^{-1} A S$, where 
\[
S:= \operatorname{diag} \left( \{\sqrt{N_j }\}_{j=1}^L \right). 
\]
Notice that by definition $B \in \mathcal A_{E_N, E_B}$. Consider the graph $\mathcal G(B) $ associated with $B$ and $\ker(B)=\ker(B^T) = \operatorname{span}(v) $ with $v_i = \sqrt{N_i} $ for every $i = 1, \dots, L$. 
We will prove that $B$ is pathwise symmetric. By Proposition \ref{prop:pdb and ps} this will imply that $A$ satisfies the pathwise detailed balance condition. 
The graph $\mathcal G(B) $ has a cut vertex $X$ by assumption. Moreover, $E_1 \subset E_B $, where $E_1 $ is as in Definition \ref{def:articulation} and $1,2 \in\mathcal G[E_1 ]$. 

As a consequence, if the walk $w \in W^{(n)}_{12}$ for some $n\geq 1 $ is such that $X \notin w$ this implies that $ \forall e \in w $, we have $ e \in  E_B$. 
The same holds for its reverse walk $w^* \in W^{(n)}_{21} $. 
Therefore, if $X \notin w$, then 
\[ 
\textbf{a}_B (w) = \textbf{a}_B (w^*). 
\]
Therefore, for every $n \geq 1 $
\[ 
\sum_{ w\in W_{12}^{(n)}, \  X \notin w } \textbf{a}_B (w) =\sum_{w^*: w\in \Pi_{12}^{(n)}, \  X \notin w  } \textbf{a}_B (w^*)  =\sum_{w\in \Pi_{21}^{(n)}, \  X \notin w } \textbf{a}_B (w). 
\]
To conclude, we prove that the same holds for walks that contain the cut vertex $X$ , hence we need to prove that 
\begin{equation}\label{eq:path with articulation}
    \sum_{w \in W_{12}^{(n)},\  X \in w} \textbf{a}_B (w)= \sum_{w \in W_{21}^{(n)}, \ X  \in w} \textbf{a}_B (w). 
\end{equation}
To this end we consider a walk $w \in W_{12}^{(n)}$ with $ X \in w$. By the definition of cut vertex, since $1, 2 \in E_1$ we have that $w$ must pass through $X$ at least two times.
Hence there exists a $(X,X)$-section of $w$ that has length larger or equal than one. In other words this $(X,X)$-section of $w$ is a cycle. Hence the walk $w $ contains at least a cycle $c_X$ such that $X \in c_X$. 

We assume, for the moment, that all the walks $w$ such that $X \in w $ of length less than $L-1$ contain only one cycle, we denote it with $c_X$, i.e.~$w=w_{(1, X) } \oplus c_X \oplus w_{(X, 2)}$. Let us define the walk $\overline w $ as  $\overline{w} =w_{(1, X) }  \oplus c_X^* \oplus  w_{(X, 2)}$ and the walk $\overline{w}^*=w^*_{(1, X) }  \oplus c_X \oplus  w^*_{(X, 2)}$, where we recall that we denote with $w^*$ the reverse of the walk $w$. Notice that $\overline w \in W_{12}^{(n)} $ and that $ w^*, \overline w^* \in W_{21}^{(n)} $. 
 
We recall that by the definition of stability class with cut vertices (Definition \ref{def:Stability class with cut vertex}) we have that every $ e \in w_{(1, X)} $ is such that $e \in E_B$ and every $ e \in w_{(X, 2)}$ is such that $ e \in E_B$. As a consequence $\textbf{a}_B( w_{(X,1)} ) = \textbf{a}_B( w_{(1, X)}) $ and $\textbf{a}_B( w_{(2, X } )=\textbf{a}_B( w_{(X,2)} )$. 
 Hence, we have that 
\begin{align*}
    \textbf{a}_B( w ) + \textbf{a}_B (\overline w ) &= \left(  \textbf{a}_B( c^*_X)  +  \textbf{a}_B( c_X) \right) \textbf{a}_B( w_{(1,X)} ) \textbf{a}_B( w_{(X,2)} ) \\
    &=  \left(  \textbf{a}_B( c^*_X)  +  \textbf{a}_B( c_X) \right) \textbf{a}_B( w_{(X,1)} ) \textbf{a}_B( w_{(2,X)} ) 
    = \textbf{a}_B( w^* ) + \textbf{a}_B (\overline w^* ).  
\end{align*}
Hence, 
\begin{align*}
    \sum_{ w \in W_{12}^{(n)}: c_X \in w } \textbf{a}_B( w ) + \sum_{ w \in W_{12}^{(n)}: c^*_X \in w } \textbf{a}_B ( w )  
    = \sum_{ w \in W_{21}^{(n)}: c^*_X \in w } \textbf{a}_B( w ) + \sum_{ w \in W_{21}^{(n)}: c_X \in w } \textbf{a}_B ( w ).  
\end{align*}
This argument can be repeated for walks crossing $X$ more than twice, namely $2m $ times for $m \geq 2 $. Implying  
 $\langle e_1, B^n e_2 \rangle = \langle e_2, B^n e_1 \rangle$ for every $n \in \mathbb N$. 
 Therefore, by the definition of the matrix $B$ this implies that the matrix $A$ satisfies the pathwise detailed balance condition.  
\end{proof}
We provide now an alternative proof of Theorem \ref{prop:stable DB}. While in the proof above we used the topological structure of the graph with a cut vertex, in the following proof we will use the structure of the matrix $A$ that corresponds to the graph with a cut vertex. 
Since the Theorem has been already proven above in this proof we consider only the case in which the matrix $A$ has steady state $N = (1, \dots ,1)/L$ and we write only the main ideas behind the proof without technical details.
\begin{proof}[Alternative proof of Theorem \ref{prop:stable DB}]
Let $X \in \Omega $ be a cut vertex. 
Let $\alpha$ and $\beta$ be the vertices of the graphs $\mathcal G[E_1]$, $\mathcal G[E_2]$ in which $X$ divides the graph $\mathcal G $.
Moreover we know that $E_1 \subset E_B $. 
Therefore, using the notation \eqref{A alpha alpha}, the matrix $E_{\alpha \cup \{X\}  } $ satisfies the detailed balance condition.
 This implies that $E_\alpha^T = E_\alpha $ and that $A_{\alpha X} = A_{X\alpha}$. 
As a consequence we can rewrite the system of ODEs \eqref{eq:ODE} as 
    \begin{align*}
   \frac{d n_\alpha }{dt} =& E_{\alpha } n_\alpha - C_{\alpha} n_\alpha +  A_{\alpha {X} } n_X   \\
   \frac{d n_X}{dt} =& A^T_{\alpha X } n_\alpha + A^T_{X \beta} n_\beta   -  \sum_{i \in \alpha} A_{i X } n_X - \sum_{i \in \beta} A_{i X } n_X  \\
   \frac{d n_\beta}{dt} =& E_{\beta } n_\beta - C_{\beta} n_\beta +  A_{\beta {X} } n_X
    \end{align*} 
where (up-on reordering) $n=(n_1, n_X, n_2)$ and where $n(0)=(n_0, 0, \textbf{0})$.

We solve the system of ODEs by performing the Laplace transform to all the terms and we obtain that for $z>0$
\begin{align*}
  \left( z - (E_\alpha - C_\alpha )  \right) \hat{n}_\alpha (z) &= n_0 + A_{\alpha X } \hat{n}_X (z) \\
  z \hat{n}_X (z) &= A^T_{\alpha X} \hat{n}_\alpha (z) +  A_{X\beta} \hat{n}_\beta  (z)-  a \hat{n}_X \\
   \left( z - (E_\beta - C_\beta )  \right) \hat{n}_\beta (z) &=  A_{\beta X } \hat{n}_X (z) \\
\end{align*}
where $a = \sum_{i \in \alpha } A_{i X} + \sum_{ i \in \beta } A_{i X} $. 

Solving these equations we deduce that for $z>0$
\begin{equation} \label{response function cut vertex}
\left( z - (E_\alpha - C_\alpha) - \frac{ (1+  A_{X \beta} \lambda(z))  A_{\alpha X} \otimes A_{\alpha X} }{z+a}\right) \hat{n}_\alpha (z) = n_0
\end{equation}
where
\[
\lambda(z):= \frac{1}{z+a} \left( z- E_\beta+ C_\beta - \frac{ A_{\beta X} \otimes A_{X\beta } }{z+a}\right)^{-1} A_{\beta X} 
\]
Now notice that the matrix 
\[
z - (E_\alpha - C_\alpha) - \frac{ (1+  A_{X \beta} \lambda(z))  A_{\alpha X} \otimes A_{\alpha X} }{z+a}
\]
is symmetric. Hence the pathwise detailed balance condition holds for every couple of vertices in $\alpha $. 
\end{proof}
The interest of this alternative proof is that it shows that the existence of a cut vertex is connected with a specific form of the response function, namely the Laplace transform of the response function satisfies \eqref{response function cut vertex}. It would be interesting to understand if all the response function of this form correspond to networks with a cut vertex. This would allow to detect the existence of a cut vertex using only the response functions.  

Theorem \ref{prop:stable DB} immediately implies that the matrices in a stability class with cut vertices satisfy the pathwise detailed balance condition in a stable way. 
\begin{corollary}\label{cor:Stability class}
    Let $ \mathcal A_{E_N, E_B}= \mathcal{C}_{E_N, E_B} $ be a stability class with cut vertices. Consider any matrix $A\in \mathcal C_{E_N, E_B}$ that satisfies the pathwise detailed balance. Then the pathwise detailed balance condition is stable in the class $\mathcal C_{E_N, E_B}$ with respect to the topology $\mathcal T(E_N, E_B)$.
\end{corollary}
In particular Corollary \ref{cor:Stability class} shows that there exists a stability class of matrices satisfying the pathwise detailed balance condition and that do not satisfy the detailed balance condition. 
So if the architecture of a graph is known and belongs to the stability class with cut vertices, then it is not possible to understand if the system satisfies the detailed balance condition having information only about the two measurements $R_{12} $ and $R_{21} $.

\begin{remark}
We stress that the assumption that $ E_1 \subset E_B $ in Theorem \ref{prop:stable DB} is necessary. 
Indeed consider the matrix 
\[
A:= \left(
\begin{matrix}
    & -3 & 1 & 1 &1 & 0 \\
    &1 &-3&1& 1 &0 \\
    &1&1 &-3 & 2 & 0 \\
    & 1 & 1 & 1 & -5 & 1 \\
    &0 & 0 & 0 & 1 & -1
\end{matrix}
\right),
\]
see Figure \ref{fig:cut vertex} for the corresponding graph. The normalized steady state of $A$ is $N= \frac{1}{24} (5,5,6,4,4)$. 
Instead the energy vector of the spanning tree with edges 
\[
\{ (1,3), (1,2), (2,4), (4,5) \} 
\]
is $\mu =\frac{1}{5}(1,1,1,1,1) $. Hence by Corollary \ref{cor:DB implies invariant measures of subtree are the same} $A$ does not satisfy the detailed balance condition. 
However $A$ satisfies the $(12)$-pathwise detailed balance.
Nevertheless, even if the graph $\mathcal G(A)$ has a cut vertex, the pathwise detailed balance property is unstable for $A$. 
    \begin{figure}[ht]
	\centering
	\includegraphics[width=0.3\linewidth]{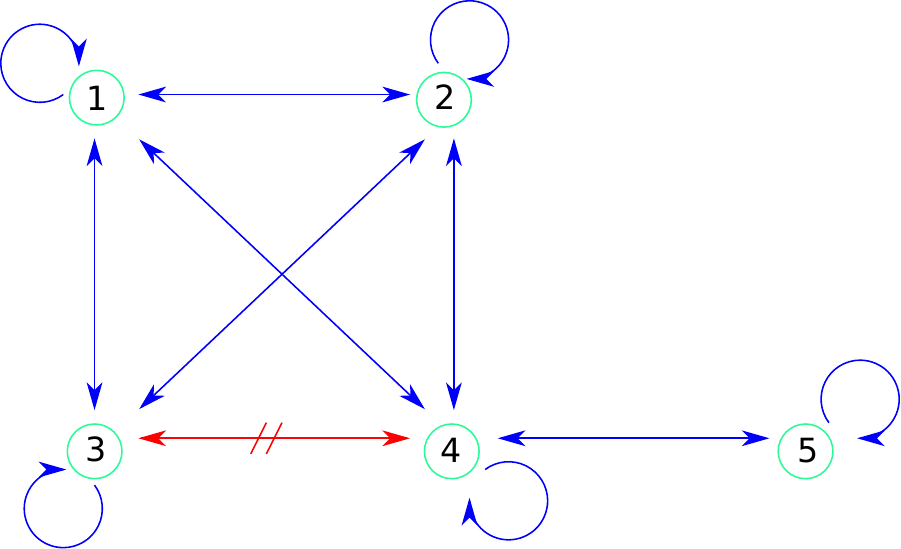}
	\caption{Graph corresponding to the matrix $A$. In red and marked with the $//$ symbol we have the only reaction without detailed balance of the network. 
 }
	\label{fig:cut vertex}
\end{figure}
\end{remark}

\section{Stable pathwise detailed balance without detailed balance implies the existence of a cut vertex} \label{sec:stable pdb implies cut vertex}
In this section we prove that the graph $\mathcal G(A) $ of a matrix $A\in \mathcal{A}$ has necessarily a cut vertex, if $A$ satisfies the pathwise detailed balance condition in a stable manner, but does not have detailed balance.

To prove this result we assume that $A$ does not satisfy the detailed balance condition, but it satisfies the pathwise detailed balance property. Moreover, we will assume that $\mathcal G(A)$ does not have a cut vertex. 
We will then construct a perturbation $D$ of the matrix $A$ such that $A + D $ does not satisfy the pathwise detailed balance condition. This implies that the pathwise detailed balance property is unstable and concludes the proof. 

To construct the aforementioned perturbation we need to prove that if $\mathcal G(V, E)$ is an undirected graph without cut vertices, there exists always a path that connects two vertices $I, F \in V $ and that contains a given vertex $\alpha=(P, Q) $. 
This will be the object of Section \ref{sec:path}.

\subsection{Construction of a suitable path in a graph without cut vertices}
\label{sec:path}
Since in this section we are interested in the topological properties of a graph we will always consider undirected graphs. 
Given $I,F \in V $ and $\alpha \in E $, let us denote with $\Pi_\alpha (I, F)$ the set of paths with origin $I$ and terminus $F$ that contain the edge $\alpha=(P,Q).$
The aim of this section is to prove that $\Pi_\alpha (I, F)\neq \emptyset$.
Notice that since the graph is undirected $\alpha = \alpha^*=(Q, P )$, hence $\Pi_\alpha (I, F) = \Pi_{\alpha^*} ( F, I)$.  

\begin{theorem} \label{thm:no articulation point, then path}
Assume $\mathcal G =(V, E) $ is an undirected connected graph that does not contain cut vertices. Let $\alpha =(P, Q) \in E $. 
For every $I, F \in V $ with  $I \neq F$, we have that 
\[
\Pi_\alpha (I, F) \neq \emptyset. 
\]
\end{theorem}
Given the vertices $X,Y$ we use the notation $\mathcal P_{ \overline V } (X,Y)$ to denote the set of the paths connecting $X$ and $Y$ that do not contain the vertices in $\overline V \subset V$. 
\begin{lemma} \label{lem:a path from X to alpha exists}
 Assume $\mathcal G =(V, E) $ is an undirected connected graph that does not contain cut vertices. Let $\alpha =(P, Q) \in E $. 
 For every $I, F \in V $ with $I \neq F$ we have that at least one of the sets
 \begin{equation} \label{four sets}
        \mathcal P_{ \{ F, Q\}  } (I, P), \, \mathcal P_{ \{ F ,P\}  }  (I, Q), \, \mathcal P_{ \{ I, Q\}  } (F, P), \, \mathcal P_{ \{ I, P\}  } (F, Q)
 \end{equation}
 is non empty.  
\end{lemma}
\begin{proof} 
Since $\mathcal G $ does not contain cut vertices the vertex $F$ is not a cut vertex. Hence the graph $\mathcal G_F:= (V_F, E_F) $ with $V_F= V \setminus \{ F\} $ and $E_F=\{ e=(e_1, e_2)  \in E : e_1, e_2  \in V_F\} $ is connected. 
By connectivity the set of the paths going from $I$ to $P$ and not containing $F$ and the set of the paths going from $I$ to $Q$ that do not contain $F$ are nonempty. 
If there exists a path $\pi$ connecting $I$ with $P$ that does not contain $F$, but contains $\alpha$, i.e.~$\pi=I \dots Q\alpha P $, then the $(I,Q)$ section of $\pi$ in $ \mathcal{P}_{\{F,P\}}(I,Q) $. Otherwise, $\mathcal P_{\{ F,Q\} }(I, P) $ is not empty.
Hence, one of the two sets $\mathcal P_{\{ F,Q\} }(I, P) $ or $\mathcal P_{\{ F,P \} }(I, Q) $ is not empty. 
\end{proof}

\begin{lemma} \label{lem:compl nonempty}
     Assume $\mathcal G =(V, E) $ is an undirected connected graph that does not contain cut vertices. Let $\alpha =(P, Q) \in E $. 
Assume that for some $I,F \in V$, with $I \neq F$ it holds that $ \mathcal P_{\{ F, Q \}} (I, P)\neq \emptyset$ and $\mathcal P_{\{ F, P \}} (I, Q)=\emptyset$.  
Then $\mathcal P_{\{ P, I \}}(Q, F)\neq \emptyset$. 
\end{lemma}
\begin{proof}
 Assume by contradiction that $\mathcal P_{\{I, P \} }  (Q, F)=\emptyset$. This implies that all the paths connecting $F $ with $Q$ must contain $P$. Indeed they cannot contain $I$ due to the assumption $\mathcal P_{\{ F, P\} }(I, Q) = \emptyset$. 
 Hence $P$ is a cut vertex. This is a contradiction, therefore $\mathcal P_{\{I, P \} }  (Q, F)\neq\emptyset$. 
\end{proof}

\begin{lemma} \label{lem:path when empty sets}
    Assume $\mathcal G =(V, E) $ is an undirected connected graph that does not contain cut vertices.  Let $\alpha =(P, Q) \in E $ and let $I, F \in V $ be such that $I\neq F$.
    Assume that  
     \begin{equation} \label{two empty alpha}
\mathcal P_{\{F, Q \} }(I, P)\neq \emptyset, \   \mathcal P_{\{F,P \} } (I, Q)= \emptyset,  \  \mathcal P_{\{I, P \} } (Q,F) \neq  \emptyset. 
    \end{equation}
    Then $\Pi_\alpha (I, F) \neq \emptyset$. 
\end{lemma}
\begin{proof}
    Assume that \eqref{two empty alpha} holds. If there exists two paths $\gamma_1 \in  \mathcal P_{\{ F, Q \} }(I, P) $ and $\gamma_2\in  \mathcal P_{\{ I, P \} }(Q, F)$ such that $\gamma_1 \cap \gamma_2=\emptyset$. 
    Then the path $\pi =\gamma_1 \oplus P\alpha Q \oplus \gamma_2$ belongs to the set $\Pi_\alpha (I, F)$. Hence the desired result follows in this case. 
    
    Assume, instead, that for every couple of paths $(\gamma_1, \gamma_2) $ with $\gamma_1 \in  \mathcal P_{\{ F, Q \}}(I, P) $ and $\gamma_2\in  \mathcal P_{\{ I, P  \}}(Q, F)$ we have that $\gamma_1 \cap \gamma_2 \neq \emptyset$. Consider a vertex $X \in \gamma_1 \cap \gamma_2 $ and define $\overline \gamma_1 $ as the $(I,X)$-section of $\gamma_1 $ and $\overline \gamma_2 $ as  the $(X, Q)$-section of $\gamma_2 $. Then $w=\overline {\gamma}_1 \oplus \overline{\gamma}_2 $ is a walk connecting $I$ with $Q$ and does not contain $P, F$. 
    Since we can extract a path $\pi $ from the walk $w$ this implies that $P_{\{F, P \} }  (I, Q)\neq \emptyset$, leading to a contradiction. Hence there exists two paths $\gamma_1 \in  \mathcal P_{\{ F, Q \} }(I, P) $ and $\gamma_2\in  \mathcal P_{\{ I, P \} }(Q, F)$ such that $\gamma_1 \cap \gamma_2=\emptyset$. As a consequence, as shown before, it follows that $\Pi_\alpha  (I, F)\neq \emptyset$. 
\end{proof}

\begin{proof}[Proof of Theorem \ref{thm:no articulation point, then path}]
 Lemma \ref{lem:a path from X to alpha exists} implies that at least one of the four sets in \eqref{four sets} is not empty. 
 Therefore, without loss of generality we can assume that $\mathcal  P_{\{ F, Q \}} (I,P) \neq \emptyset$. 
 Now we have two possibilities, either $\mathcal  P_{\{ F, P \}} (I,Q) \neq \emptyset$ or $\mathcal  P_{\{ F, P \}} (I,Q) = \emptyset$. 
 If $\mathcal  P_{\{ F, P \}} (I,Q) = \emptyset$ holds, then Lemma \ref{lem:compl nonempty} implies that $\mathcal P_{\{ I, P \} } (Q,F) \neq \emptyset$. 
Hence Lemma \ref{lem:path when empty sets} implies that $\Pi_\alpha (I, F) \neq \emptyset$. 
If instead we have that $\mathcal  P_{\{ F, P \}} (I,Q) \neq \emptyset$, then we have three options: 
\begin{enumerate}
    \item either $\mathcal P_{\{I, P \} } (F, Q ), \, \mathcal P_{\{I, Q \} } (F, P ) \neq \emptyset $
    \item or $\mathcal P_{\{I, P \} } (F, Q ) = \emptyset, \,  \mathcal P_{\{I, Q \} } (F, P ) \neq \emptyset $
    \item or $\mathcal P_{\{I, P \} } (F, Q ) \neq \emptyset, \, \mathcal P_{\{I, Q \} } (F, P ) = \emptyset $.
\end{enumerate}
If one of the two last cases occurs, we can repeat the argument above relabelling $(I,P ) $ with $(F,Q) $ in case (c) and $(I,P) $ with $(F, P)$ in case (b). In both cases we deduce that $\Pi_\alpha (I, F) \neq \emptyset$. It remains to study the first case, in which all the sets \eqref{four sets} are non empty. 

Therefore we assume from now on that all the sets \eqref{four sets} are non empty. 
We use the following notation
\[
\mathcal P(I, F; P, Q) := \mathcal P_{\{ F, Q  \} }(I,P) \cup \mathcal P_{\{ F, P  \} } (I,Q) \cup  \mathcal P_{\{ I, P \} } (Q,F) \cup  \mathcal P_{\{ I, Q  \} }(P, F ). 
\]
Let 
   \[
   D(I,F):= \min_{p \in \mathcal P(I, F; P, Q) } \ell(p) 
   \]
We aim at proving that $\Pi_\alpha  (I, F) \neq \emptyset$ by induction over $D(I,F)$. 

   Assume that $D(I,F)=1$. 
   Then there exists a path $\pi_1$ of length $1$ in the set $ \mathcal P(I,F; P, Q) $. 
   Without loss of generality we assume that this path $\pi_1$ belongs to the set $\mathcal P_{\{ F, Q\}}(I,P)$. Hence, since $\ell(\pi_1)=1$ we have that $\pi_1=I (I,P) P.$
   By assumption we have that $\mathcal P_{\{ I,P\}} (Q,F) \neq \emptyset$. 
  Hence we can consider a $ \pi_2 \in \mathcal P_{\{ I,P\}} (Q,F)$. The path $ \pi_2$ does not contain the edge $\alpha $ and does not contain the vertices $I,P$. 
   Hence $\pi_1 \oplus P \alpha Q \oplus \pi_2$ is a path from $I$ to $F$ that contains $\alpha $.
   Hence $\Pi_\alpha  (I, F) \neq \emptyset$.

We now assume that $D(I, F; P, Q) = d+1 $ with $d \geq 1 $.
Hence there exists a path $\pi \in \mathcal P(I, F) $ such that $\ell(\pi)=d+1$. 
Without loss of generality, we assume that $\pi \in \mathcal P_{\{ F, Q \}}(I, P)$. As a consequence there exists a $X \in V $, with $X \neq P$, such that $\pi=I (I, X) X \dots P$. Notice that, since $\pi \in \mathcal P_{\{ F, Q \}}(I, P)$ we have that $X \neq F $, $X \neq Q$ and $X \neq I$.
Notice, that by construction $\mathcal P_{ \{ F, Q \} } (X, P) \neq \emptyset$, and therefore $\mathcal P (X, F ; P, Q) \neq \emptyset$. As a consequence we can compute $D(X, F) $ and we have that $D(X, F)=d $ because the $(X, P)$-section  of $\pi $ has length $d$ by construction and belongs to $\mathcal P(X, F; P, Q)$. 
The induction hypothesis implies that the set $\Pi_\alpha (X, F ) \neq \emptyset$. 
We have two possibilities
\begin{enumerate}
    \item $\forall \gamma \in \Pi_\alpha (X, F ) $ we have that $I \in \gamma $;
    \item $\exists \gamma \in \Pi_\alpha  (X, F ) $ such that $I \notin \gamma $. 
\end{enumerate}
If $(b) $ holds then $I(I,X)X \oplus \gamma$ is a path that belongs to the set $\Pi_\alpha (I, F)$ and the desired result follows. 

Assume now that $(a) $ holds. Then we have the following four possibilities 
\begin{enumerate}[(1)]
    \item there exists a $\gamma $ of the form $\gamma =\gamma_1 \oplus \gamma_2 \oplus P \alpha  Q  \oplus  \gamma_3 $ with $ \gamma_1 \in \mathcal P_{\{ P, Q, F\}}(X, I )$,  $ \gamma_2 \in \mathcal P_{\{ X, Q, F\}}( I, P  )$ and $\gamma_3 \in \mathcal P_{\{ P, I, X \}}( Q, F)$; 
      \item there exists a $\gamma $ of the form $\gamma =\gamma_1 \oplus \gamma_2 \oplus Q\alpha P \oplus  \gamma_3 $ with $ \gamma_1 \in \mathcal P_{\{ P, Q, F\}}(X, I )$, $ \gamma_2 \in \mathcal P_{\{ X, P, F\}}( I, Q  )$ and $\gamma_3 \in \mathcal P_{\{ Q, I, X \}}( P, F)$; 
                   \item there exists a $\gamma $ of the form $\gamma =\gamma_1  \oplus P \alpha Q  \oplus \gamma_2\oplus  \gamma_3 $ with $\gamma_1 \in \mathcal P_{\{I, F, Q \}}(X, P )$, $\gamma_2 \in \mathcal P_{\{ F, P, X\}  } (Q, I )$ and $\gamma_3 \in \mathcal P_{ \{ P, Q, X \} } (I, F )$;
         \item there exists a $\gamma $ of the form $\gamma =\gamma_1 \oplus  Q \alpha  P  \oplus \gamma_2 \oplus \gamma_3 $ with $\gamma_1 \in \mathcal P_{\{ P, I, F \} }  P ( X, Q )$, $\gamma_2 \in \mathcal P_{\{ Q, X, F\}} (P, I )$ and $\gamma_3 \in \mathcal P_{\{ X, P, Q\} }  P (I, F )$.
\end{enumerate}

\textbf{Analysis of the case (1) and (2).} 
Assume that $1$ holds. Then $\gamma_2 \oplus P\alpha Q  \oplus \gamma_3 \in \Pi_\alpha (I,F)$ and the desired conclusion follows. 
Similarly, if $2$ holds then $\gamma_2 \oplus Q \alpha  P  \oplus \gamma_3 \in \Pi_\alpha  (I, F)$.

\textbf{Analysis of the case (3).} 
Assume now that $3 $ holds. Then we have two possibilities 
\begin{enumerate}
    \item $\exists \beta \in \mathcal P_{\{ X, P \} } (Q, F) $ such that $ I \notin \beta$; 
    \item $\forall \beta \in \mathcal P_{\{ X, P \} } (Q, F) $ we have that $I \in \beta$. 
\end{enumerate}
Assume that $(a)$ holds. 
Then we analyse all the possible intersections of $\gamma $ with $\beta$. 
\begin{itemize}
    \item If $ \beta \cap \gamma_1 =\emptyset$, then $I (I, X) X  \oplus \gamma_1 \oplus P\alpha Q \oplus \beta \in \Pi_\alpha (I, F)$. 
     \item  If $\beta \cap \gamma_1 \neq \emptyset$ and $\beta \cap \gamma_2 =\emptyset$, then there exists a $Y \in \gamma_1 \cap \beta$ such that
     \begin{align} \label{minimality1}
         \operatorname{dist}_\beta (F, Y) = \min \{ \operatorname{dist}_\beta (F, Z) \, : \, Z\in \gamma_1\cap \beta\}.
     \end{align}
     We stress that here we are using the notation \eqref{distance between wertices on a path} for $\operatorname{dist}_\beta (F, Z)$.
     We obtain that $\gamma_2 \oplus Q \alpha P \oplus {\gamma_1}_{(P,Y)} \oplus \beta_{(Y,F)} \in \Pi_\alpha  (I, F)$. Note that $\beta_{(Y,F)} \cap {\gamma_1}_{(P,Y)} = \{Y\} $ by \eqref{minimality1}.
    \item  If $\beta \cap \gamma_1 \neq \emptyset$ and $\beta \cap \gamma_2\neq \emptyset $, then there exists a $Y_1 \in \gamma_1 \cap \beta$ and a $Y_2 \in \gamma_2 \cap \beta$ such that
\begin{align*}
    \operatorname{dist}_\beta (F, Y_1) &= \min \{ \operatorname{dist}_\beta (F, Z) \, : \, Z\in \gamma_1\cap \beta\},
\\
\operatorname{dist}_\beta (F, Y_2) &= \min \{ \operatorname{dist}_\beta (F, Z) \, : \, Z\in \gamma_2\cap \beta\}.
\end{align*}
Notice that since $\gamma $ is a path $\gamma_1 \cap \gamma_2 = \emptyset$. Then $Y_1 \neq Y_2$.
If $\operatorname{dist}_\beta (F, Y_2 ) > \operatorname{dist}_\beta (F, Y_1 )  $, then $ \gamma_2 \oplus Q \alpha P  \oplus {\gamma_1}_{(P, Y_1)} \oplus {\beta}_{( Y_1, F) } \in \Pi_\alpha (I, F)$. If instead $\operatorname{dist}_\beta (F, Y_2 ) < \operatorname{dist}_\beta (F, Y_1 )  $, then $ I(I,X) X \oplus \gamma_1 \oplus P\alpha Q  \oplus {\gamma_2}_{(Q, Y_2)} \oplus \beta_{(Y_2, F) }\in \Pi_\alpha  (I, F)$.  
\end{itemize}
In all the possible cases we have $\Pi_\alpha  (I, F)\neq \emptyset $ , hence if $(a) $ holds the result follows. 

Let us assume now that $(b) $ holds. We prove now that $(b)$ implies that for every $\eta \in \mathcal P_{\{I,P\} } (Q, F)$ we have that $X \in \eta$.  Hence  $\gamma_1 \cap \eta \neq \emptyset$.
This can be proven by contradiction. Assume that there exists $\eta \in \mathcal P_{\{ I, P \} } (Q, F)$ such that $X \notin \eta $. This implies that $\eta \in  \mathcal P_{ \{I, P, X\} } (Q, F) \subset \mathcal P_{ \{ P, X \} } (Q, F) $. Hence assumption $(b) $ implies $I \in \eta $, which is a contradiction. Therefore we have that every $\eta \in \mathcal P_{ \{I, P\} } (Q, F)$ is such that $X \in \eta$. 
In particular this implies that $\gamma_1 \cap \eta \neq \emptyset$. 
Also in this cases we have two possibilities.  
\begin{itemize}
    \item If there exists a $\eta \in \mathcal P_{\{I, P\} } (Q, F)$ such that $ \eta \cap \gamma_1 \neq \emptyset$ and $\eta \cap \gamma_2 =\emptyset$. 
    Let $ Y \in \gamma_1 \cap \eta $ be such that
      \begin{align} \label{minimality2}
         \operatorname{dist}_\eta (F, Y) = \min \{ \operatorname{dist}_\eta (F, Z) \, : \, Z\in \gamma_1\cap \eta\}.
     \end{align}   
    then $\gamma_2 \oplus Q\alpha P \oplus {\gamma_1}_{(P, Y) } \oplus \eta_{(Y, F) } \in\Pi_\alpha  (I,F)$. 
    \item  Assume every $\eta \in \mathcal P_{\{I, P\} } (Q, F)$ is such that $\eta \cap \gamma_1 \neq \emptyset$ and $\eta \cap \gamma_2 \neq \emptyset$. 
    In this case consider
     one of the paths $\eta \in  \mathcal P_{ \{ I, P \} } (Q, F)$.    Let $Y_1 \in \eta \cap \gamma_1 $ and $Y_2 \in \eta \cap \gamma_2 $
         \begin{align*}
    \operatorname{dist}_\eta (F, Y_1) &= \min \{ \operatorname{dist}_\eta (F, Z) \, : \, Z\in \gamma_1\cap \eta\},
\\
\operatorname{dist}_\eta (F, Y_2) &= \min \{ \operatorname{dist}_\eta (F, Z) \, : \, Z\in \gamma_2\cap \eta\}.
\end{align*}
Notice that since $\gamma $ is a path $\gamma_1 \cap \gamma_2 = \emptyset$. Then $Y_1 \neq Y_2$.
  If $\operatorname{dist}_{\eta} (F, Y_2 ) < \operatorname{dist}_{\eta} (F, Y_1 ) $,
    then $I (I, X) X \oplus \gamma_1 \oplus P \alpha Q \oplus {\gamma_2}_{(Q, Y_2) } \oplus \eta_{(Y_2, F )} \in \Pi_\alpha  (I,F)$. 
    If instead $\operatorname{dist}_{\eta} (F, Y_2 ) > \operatorname{dist}_{\eta} (F, Y_1 ) $, then $ \gamma_2 \oplus Q \alpha P \oplus {\gamma_1}_{(P, Y_1)} \oplus \eta_{(Y_1,F)} \in \Pi_\alpha (I,F)$.  
\end{itemize}
Hence $(b) $ implies that $\Pi_\alpha (I, F) \neq \emptyset . $
Hence the desired conclusion follows if assumption (3) holds.

\textbf{Analysis of the case (4).} This case can be analysed as the case $(3) $ relabelling $P$ with $Q$.
\end{proof}

\subsection{Stable pathwise detailed balance implies the existence of a cut vertex}
\label{sec:theorem cut vertex}

In this section we prove that stable pathwise detailed balance implies the existence of a cut vertex. 

\begin{theorem} \label{thm:stable pdb implies articulation}
Let $A \in \mathcal A_{E_N, E_B} $ and let $\mathcal G =\mathcal G(A)$ be the graph associated to $A$, where $E_N,\, E_B\subset \mathcal{S}$ are as in Definition \ref{def:class}. Assume that $A$ satisfies the pathwise detailed balance condition and assume that this property is stable in $\mathcal A_{E_N, E_B}$ with respect to the topology $ \mathcal{T}(E_N,E_B) $. Then we have two possibilities
\begin{enumerate}
    \item either the matrix $A$ satisfies the detailed balance condition; 
    \item or the graph $\mathcal G $ has a cut vertex. 
\end{enumerate}
\end{theorem}
\begin{remark}
    Notice that in Theorem \ref{thm:stable pdb implies articulation} we allow perturbations merely in the stronger topology $\mathcal{T} (E_N,E_B)$. 
    As a consequence when we can also consider perturbations of the matrix $A$ in the weaker topology $\mathcal{T}_w(E_B)$.
\end{remark}
Before writing the proof of this theorem we give an informal explanation of the proof. We assume that there exists a matrix $A$ that satisfies the pathwise detailed balance condition in a stable manner and is such that the graph induced by $A$, $\mathcal G(A)$ does not have a cut vertex.  We prove that this implies that $A$ satisfies the detailed balance condition. 

We argue by contradiction. Assume that there exists an edge $\alpha $ on which the detailed balance condition does not hold. Since $A$ satisfies the pathwise detailed balance condition $\Delta_n (A)=0$ in a stable manner, we have $\Delta_n(\overline A) =0$ for $\overline A = A+ D$, where $D$ is any admissible, sufficiently small perturbation. However, we prove that $\Delta_n(\overline A) \neq0$ for a certain choice of $D$, leading to a contradiction. We now explain how we construct this suitable perturbation $D$.

Notice that the sums in $\Delta_n(\overline A)$ involve positive and negative terms. It is therefore not easy to find a perturbation avoiding cancellations. To this end, we perturb all the rates of the matrix $A$ that correspond to edges, different from $\alpha$, of a specific path $\pi\in \Pi_\alpha (1, 2)$. Theorem \ref{thm:no articulation point, then path} guarantees the existence of such a path. Furthermore, we choose the perturbation of these edges in such a way that the matrix $D$ has steady state $N$ and satisfies the detailed balance condition with respect to $N$.

Differentiating $\Delta_n(\overline A) $ with respect to the non-zero entries of $D$ allows to reduce the sums in $\Delta_n(\overline A) $  to a single term. This term turns out to be non-zero due to the fact that in $\alpha$ the detailed balance condition fails. This implies that $\Delta_n (\overline A) \neq 0 $ leading to a contradiction.

\begin{proof}[Proof of Theorem \ref{thm:stable pdb implies articulation}]
Consider a matrix $A \in \mathcal A_{E_N, E_B}$ and assume that the pathwise detailed balance condition holds in a stable manner in $\mathcal A_{E_N, E_B}$. Furthermore, assume that the graph $\mathcal G$ does not have a cut vertex. We want to prove that this implies that $A$ satisfies the detailed balance condition. 
To this end we proceed by contradiction and assume that there exists an edge $\alpha =(P, Q )$ of the graph $\mathcal G $ that does not satisfy the detailed balance condition. In other words we have that $A_{PQ } N_Q \neq A_{Q P} N_P  $, where $N$ is the normalized steady state of $A$. 

Since the graph $\mathcal{G} $ does not have cut vertices, Theorem \ref{thm:no articulation point, then path} implies that for every $I, F \in \Omega$ we have that $\Pi_\alpha  (I, F) \neq \emptyset$. 
Let us consider a path $\pi  \in \Pi_\alpha (1, 2) $ for some $1,2 \in \Omega$ with $1 \neq 2 $.
We aim at perturbing all the rates of the edges that belong to the path $\pi$. 

Given an edge $e=(e_1, e_2) \in \pi $ such that $e \neq \alpha $, we define the matrix $D(e)$ as
\begin{align}\label{eq:D DB}
    \begin{split}
        D_{ ij}(e)  &=0 \quad \forall (i, j ) \neq e,e^*, (e_1,e_1), (e_2, e_2),
        \\
        D_{e_1 e_2 }(e) &= \frac{N_{e_1}}{N_{e_2}}, \ D_{e_2, e_1}(e)=1,\ D_{e_1 e_1} = - \frac{N_{e_1}}{N_{e_2}}, \ D_{e_2 e_2 }(e) =- 1.
    \end{split}
\end{align}
Notice that $\textbf{e}^T_L D(e)=0$ and that $D(e)N =0$. Therefore, also in this case $N$ is the unique steady state of $\overline A(e)= A + \varepsilon (e) D(e)$, where $\varepsilon (e) >0$. 
Moreover $\overline A(e) \in \mathcal A_{E_N, E_B}$. 

Let 
\[
D(\pi ):= \sum_{e \in \pi : e \neq \alpha } \varepsilon (e) D(e). 
\]
 
We define the perturbation $\overline A (\pi)$ of $A$ on $\pi \setminus \{ \alpha \}$ as 
\begin{equation}\label{eq:perturbed A any set}
\overline A (\pi)= A +  D(\pi). 
\end{equation}
 By construction $N$ is the unique steady state of $\overline A(\pi)$ and $\overline A(\pi) \in \mathcal A_{E_N, E_B}$.

Let $n:= \ell (\pi)$.  
Since $A$ satisfies the pathwise detailed balance condition we know that 
$\Delta_n(A) =0$, where $\Delta_n(A) $ is defined as in \eqref{differential}. 
 The fact that the pathwise detailed balance condition is stable in $\mathcal{A}_{E_N, E_B}$ implies that also the matrix $\overline A(e) \in \mathcal {A}_{E_N, E_B}$ satisfies the pathwise detailed balance condition, hence $\Delta_n( \overline A(\pi)) =0$, if all $\varepsilon (e) $ in the definition of $D(\pi)$ are sufficiently small.

 On the other hand,
 \begin{align} \label{delta A bar}
 \Delta_n (\overline A(\pi )) =  N_1 \sum_{w \in W_{12}^{(n)}} \textbf{a}_{\overline A(\pi )} (w) - N_2\sum_{w \in W^{(n)}_{21} } \textbf{a}_{\overline A(\pi ) } (w) 
 \end{align}
 where $W_{12}^{(n)}$ is the set of the walks from $1$ to $2$ of length $n$ and $\textbf{a}_{\overline A(\pi )  } (w) $ is given by \eqref{rate of walk}.  

Now notice that
\begin{align} \label{a(w)}
\textbf{a}_{\overline A(\pi )} (w) &= \prod_{ e=(e_1, e_2) \in w}  \overline{A}^{e_1}_{ e_2} (\pi ) = \prod_{ e=(e_1, e_2) \in w}  \left(A^{e_1}_{ e_2} + D^{e_1}_{e_2} (\pi ) \right) \nonumber \\
&= \prod_{ e=(e_1, e_2) \in w}  A^{e_1}_{ e_2} +  \prod_{ e=(e_1, e_2) \in w}  D^{e_1}_{e_2} (\pi ) + \sum_{E(w) \subset w } \   \prod_{ e \in E(w) }  D^{e_1}_{e_2} (\pi ) \prod_{ e \in w \setminus  E(w) }  A^{e_1}_{e_2}. 
\end{align}
Here we are using the notation $E(w) \subset w $ to say that $E(w) $ is a subset of the set of the edges of $w$. 

We now compute $\textbf{a}_{\overline A(\pi )} ( \pi ) $ applying \eqref{a(w)} to $w = \pi $
\begin{align*}
\textbf{a}_{\overline A(\pi )} (\pi)  &= \prod_{ e=(e_1, e_2) \in \pi }  A^{e_1}_{ e_2} +  \prod_{ e=(e_1, e_2) \in \pi }  D^{e_1}_{e_2} (\pi ) + \sum_{E(\pi) \subset \pi  } \ \prod_{ e \in E(\pi ) }  D^{e_1}_{e_2} (\pi ) \prod_{ e \in \pi  \setminus  E(\pi ) }  A^{e_1}_{e_2} \\
&=   \prod_{ e=(e_1, e_2) \in \pi }  A^{e_1}_{ e_2}  + \sum_{E(\pi) \subset \pi  }  \ \prod_{ e \in E(\pi ) }  D^{e_1}_{e_2} (\pi ) \prod_{ e \in \pi  \setminus  E(\pi ) }  A^{e_1}_{e_2}. 
\end{align*}
The last equality is due to the fact that $\prod_{ e=(e_1, e_2) \in \pi }  D^{e_1}_{e_2} (\pi )=0$ because $D^{\alpha_1}_{\alpha_2 }(\pi) = D^{\alpha_2}_{\alpha_1}(\pi) =0$. 
Moreover we can rewrite $\textbf{a}_{\overline A(\pi )} (\pi)$ as 
\begin{align} \label{a(pi)}
\textbf{a}_{\overline A(\pi )} (\pi ) =\textbf{a}_{ A} (\pi ) + A^{\alpha_1}_{\alpha_2} \prod_{ e \in \pi :  e \neq \alpha  }  D^{e_1}_{e_2} (\pi ) + \sum_{\{ E(\pi ) \subset \pi : E(\pi) \neq \pi \setminus \{ \alpha \}  \} } \quad  \prod_{ e \in E(\pi ) }  D^{e_1}_{e_2} (\pi ) \prod_{ e \in \pi  \setminus  E(\pi ) }  A^{e_1}_{e_2}. 
\end{align}

Similarly, 
\begin{align} \label{a(pi*)}
\textbf{a}_{\overline A(\pi )} (\pi^* ) =\textbf{a}_{ A} (\pi^* ) + A^{\alpha_2}_{\alpha_1} \prod_{ e \in \pi^* :  e \neq \alpha^*  }  D^{e_1}_{e_2} (\pi ) + \sum_{\{ E(\pi^* ) \subset \pi : E(\pi^*) \neq \pi \setminus \{ \alpha^* \} \} }  \prod_{ e \in E(\pi^* ) }  D^{e_1}_{e_2} (\pi ) \prod_{ e \in \pi  \setminus  E(\pi^* ) }  A^{e_1}_{e_2}. 
\end{align}

Now notice that, since $D(\pi) $ satisfies the detailed balance condition with respect to the steady state $N$ we have that 
\begin{align*}
N_1 A^{\alpha_1}_{\alpha_2} \prod_{ e \in \pi :  e \neq \alpha  }  D^{e_1}_{e_2} (\pi )= A^{\alpha_1}_{\alpha_2} \frac{N_{\alpha_1} N_2  }{N_{\alpha_2} }   \prod_{ e \in \pi^* :  e \neq \alpha^* }  D^{e_1}_{e_2} (\pi).
\end{align*}
Therefore 
\begin{align*}
N_1 A^{\alpha_1}_{\alpha_2} \prod_{ e \in \pi :  e \neq \alpha  }  D^{e_1}_{e_2} (\pi )- N_2 A^{\alpha_1}_{\alpha_2} \prod_{ e \in \pi^* :  e \neq \alpha  }  D^{e_1}_{e_2} (\pi )    = N_2 \prod_{ e \in \pi^* :  e \neq \alpha^* }  D^{e_1}_{e_2} (\pi) \left( \frac{N_{\alpha_1}  }{N_{\alpha_2} }   A^{\alpha_1}_{\alpha_2}-  A^{\alpha_2}_{\alpha_1} \right) \neq 0 . 
\end{align*}
As a consequence we notice that 
\begin{equation} \label{nonzero term}
\left( \prod_{e \in \pi : \alpha \neq e } \frac{d}{d \varepsilon (e) } \right) \left( N_1 \textbf{a}_{\overline A(\pi )} (\pi ) - N_2 \textbf{a}_{\overline A(\pi )} (\pi^* ) \right) = N_2  \left( \frac{N_{\alpha_1}  }{N_{\alpha_2} }   A^{\alpha_1}_{\alpha_2}-  A^{\alpha_2}_{\alpha_1} \right) \neq 0. 
\end{equation} 
Indeed we have that 
\[
\left( \prod_{e \in \pi : \alpha \neq e } \frac{d}{d \varepsilon (e) } \right) \textbf{a}_{ A} (\pi^* )= \left(  \prod_{e \in \pi : \alpha \neq e } \frac{d}{d \varepsilon (e) } \right)  \textbf{a}_{ A} (\pi)=0
\]
as well as 
\[
\left( \prod_{e \in \pi : \alpha \neq e } \frac{d}{d \varepsilon (e) } \right) \sum_{\{ E(\pi ) \subset \pi : E(\pi) \neq \pi \setminus \{ \alpha \} \} } \ \prod_{ e \in E(\pi ) }  D^{e_1}_{e_2} (\pi ) \prod_{ e \in \pi  \setminus  E(\pi ) }  A^{e_1}_{e_2}  =0
\]
and 
\[
\left( \prod_{e \in \pi : \alpha \neq e } \frac{d}{d \varepsilon (e) } \right) \sum_{\{ E(\pi^* ) \subset \pi^* : E(\pi^*) \neq \pi^* \setminus \{ \alpha^* \} \} } \ \prod_{ e \in E(\pi^* ) }  D^{e_1}_{e_2} (\pi^* ) \prod_{ e \in \pi ^* \setminus  E(\pi^* ) }  A^{e_1}_{e_2}  =0. 
\]
Notice that here we are using the important assumption that $\pi $ is a path, hence $\alpha^* \notin \pi $, since $\alpha \in \pi$. 
This is the reason why we need Theorem \ref{thm:no articulation point, then path}. 

Now notice that for every $w \in W_{12}^{(n)} $, $w \neq \pi $ there exists $e \in \pi \setminus \{ \alpha \} $ such that $ e \notin w $.
Therefore we have first of all that, if $w \neq \pi $, then 
\[
  \prod_{ e \in w  }  D^{e_1}_{e_2} (\pi ) =0, 
\]
as well as 
\[
\left( \prod_{e \in \pi : \alpha \neq e } \frac{d}{d \varepsilon (e) } \right) \sum_{E(w) \subset w }  \ \prod_{ e \in E(w) }  D^{e_1}_{e_2} (\pi ) \prod_{ e \in w \setminus  E(w) }  A^{e_1}_{e_2}=0. 
\]
Finally since $ \textbf{a}_{ A} (w)$ does not depend on $D(\pi) $ we also have that
\[
\left( \prod_{e \in \pi : \alpha \neq e } \frac{d}{d \varepsilon (e) } \right) \textbf{a}_{ A} (w)=0. 
\]

The above equalities, together with \eqref{delta A bar}, \eqref{a(pi)}, \eqref{a(pi*)}, \eqref{nonzero term} and \eqref{a(w)}, that 
\begin{align*}
\left( \prod_{e \in \pi : \alpha \neq e } \frac{d}{d \varepsilon (e) } \right) \Delta_n (\overline A) &=
\left( \prod_{e \in \pi : \alpha \neq e } \frac{d}{d \varepsilon (e) } \right) \left( N_1 \textbf{a}_{\overline A(\pi )} (\pi ) - N_2 \textbf{a}_{\overline A(\pi )} (\pi^* ) \right) \\
&= N_2  \left( \frac{N_{\alpha_1}  }{N_{\alpha_2} }   A^{\alpha_1}_{\alpha_2}-  A^{\alpha_2}_{\alpha_1} \right) \neq 0. 
\end{align*}
This implies that $\Delta_n (\overline A) \neq 0$ for sufficiently small $\varepsilon (e) =0 $ for every $e \in \pi \setminus \{ \alpha \}  $. This contradicts $\Delta_n (\overline A) = 0$ and concludes the proof.
\end{proof}
Theorem \ref{thm:stable pdb implies articulation} implies that, if a matrix does not have a cut vertex and satisfy pathwise detailed balance in a stable manner, then it satisfies the detailed balance condition. Hence we have the following corollary of Theorem \ref{thm:stable pdb implies articulation}. 
 
\begin{corollary} \label{cor:unstable PDB complete graph with one detailed balance}
Let $L>3$. Assume that the matrix $A \in \mathbb R^{L \times L } $ belongs to the class $\mathcal A_{\{ \emptyset, \{(1,2), (2,1)\} \} } $. Assume that $A$ satisfies the pathwise detailed balance condition in a stable manner with respect to the topology $\mathcal{T}(\emptyset,\{(1,2), (2,1)\} )$. Then, $A$ satisfies the detailed balance condition.
\end{corollary}
\begin{proof}
    The graph $\mathcal G(A)$ is complete because $E_N = \emptyset$, hence it does not contain a cut vertex. Theorem~\ref{thm:stable pdb implies articulation} implies the statement.
\end{proof}

\section{Approximating the response functions using the time series of the Markov process associated to the biochemical network}\label{sec:Stochastic}
A possible way to measure the response function $R_{12} (t) $ is to inject substance $1$ in the system and measure substance $2$. If we have an homogeneous system of cells this would give us $R_{12}$. 
On the other hand, one could imagine other types of experiments, where the observed quantity is a time series of events. 
 In the case of chemical systems modelled by equations with the form \eqref{eq:ODE}, a natural time series is the sequence of the times in which some of the states of the system take palace, say $1$ and $2$. In this section we will indicate how to detect the property \eqref{eq:intro p(DB)} in such time series. Specifically, we explain how to identify \eqref{eq:intro p(DB)} in a single realization of the (hidden) Markov process associated to \eqref{eq:ODE} with probability one.

First of all, we recall that a chemical network given by ODEs of the form \eqref{eq:ODE}, with a Markovian matrix $ A $, is associated with a pure-jump Markov process $ \{ X_t \}_{ t \geq 0} $ with transition probability rate $A_{ji}dt$ from state $i$ to $j$ in the time interval $[t,t+dt]$. 

The response function $R_{21}$ is related to the process $ (X_t)_{t\geq0} $, initialized at time $t=0$ as $X_0$ with probability density $ n_0=(0,1,0,\ldots,0) $, by the formula
\begin{align*}
    \P(X_t=1\mid X_0=2) = n_1(t) =  R_{21}(t).
\end{align*}
Recall that the Markov process is translation invariant with respect to time, in particular we have
\begin{align*}
    \P(X_{t_2}=1\mid X_{t_1}=2) = R_{21}(t_2-t_1), \quad t_2>t_1\geq0.
\end{align*}
We recall that the stability of the property \eqref{eq:intro p(DB)} implies detailed balance condition for systems without cut vertices, see Theorem \ref{thm:stable pdb implies articulation}. 

In principle, to compute the response functions $R_{12} $, $R_{21} $, would require a detailed knowledge of the underlying Markov process $X_t$, or in other words of the rates of all the reactions taking place in the biochemical network. On the other hand, the two functions $R_{12} (t) $ and $R_{21}(t)$ can be computed in several different ways by means of realizations of stochastic processes. We first consider the case in which, to reconstruct the response functions, we use multiple independent realizations of the process and apply the Law of Large Numbers, (see e.g. \cite{feller1991introduction}). 
\begin{theorem} \label{thm:stoch}
	Consider infinitely many independent measurements $X_t^k$ and $\bar X_t^k$, $k\in \N$, of the process $(X_t)_{t\geq0}$ with $X_0=1$ and of the process $(\bar X_t)_{t\geq0}$ with $\bar X_0=2$, respectively. Then we have
    \begin{align*}
	    \lim_{N\to \infty} \dfrac{1}{N} \sum_{k=1}^N \ind_{\{X_t^k = 2\}} = R_{12}(t),
        \quad 
        \lim_{N\to \infty} \dfrac{1}{N} \sum_{k=1}^N \ind_{\{\bar X_t^k = 2\}} = R_{21}(t).
\end{align*}
\end{theorem}
The above theorem allows us to check if the condition $R_{12}(t)/R_{21}(t)=const. $, cf. \eqref{eq:intro p(DB)}, holds.

Let us now explain how we can, alternatively, verify condition \eqref{eq:intro p(DB)} using a single realization of the process using the ergodicity assumption. Let us fix two times $t_2>t_1>0$. Consider the process $(X_t)_{t\geq0}$ starting at $X_0=1$. We then make the following experiments $N$ times leading to an approximation of $R_{12}(t)$.
\begin{enumerate}[(1)]
    \item Wait for time $t_1$ and measure if $X_{t_1}=2$ or $X_{t_1}\neq2$.
    \item Wait until time $t_2>t_1$ and measure if $ X_{t_2}=2$ or $X_{t_2}\neq2$.
    \item Wait now until the process reaches the state $1$. Due to the ergodicity of the process this will happen in finite time, say $\tau$.
    \item Restart the process at time $t_2+\tau$ and go back to Step (1).
\end{enumerate}
Let us mention that this yields $N$ independent measurements $(X^k_{t_1})_{k=1}^N$, $(X^k_{t_2})_{k=1}^N$ of the process $(X_t)_{t\geq0}$ starting in $X_0=1$ due to the strong Markov property. Accordingly, in the same realization we can wait until the process reaches the state $2$. Again, this happens in finite time due to the ergodicity. Then, we can repeat the above procedure exchanging $1$ and $2$. This yields $N$ independent measurements $(\bar X^k_{t_1})_{k=1}^N$, $(\bar X^k_{t_2})_{k=1}^N$ of the process $(\bar X_t)_{t\geq0}$ starting in $\bar X_0=2$. Using Theorem \ref{thm:stoch} this yields approximations of the response functions $R_{12}(t), \, R_{21}(t) $ for $t=t_1,\, t_2$. If the equality $R_{12}(t_1)/R_{21}(t_1) = R_{12}(t_2)/R_{21}(t_2) $ fails, we deduce that detailed balance fails. Otherwise, we repeat the experiments for different times. Notice that with a sufficiently large number of times $t_1, t_2$ we can approximate the response functions with arbitrary accuracy.

The use of time series suggested above requires very long measurements. 
Most likely there  is a way of optimizing the information in a time series, in order to obtain $R_{12}(t) $ and $R_{21}(t) $. However we are not going to study this in this paper. 

\section{Extended detailed balance condition} \label{sec:extended db}
All the results in the previous sections are for Markovian matrices $A$. From the biological point of view this means that we restrict our attention to chemical systems that are isolated, because we assume that the total concentration of substances in the system does not change in time. 

In this section we illustrate how it is possible to extend the previous results to the case of chemical systems where the number of substance in the population might vary, due to the presence of sources and sinks. 
In particular, we extend the definition of detailed balance to chemical networks with sources and sinks and we prove that, if a matrix $A$ satisfies the extended detailed balance condition, then it satisfies also the pathwise detailed balance condition. 

To this end we adopt the following notation in \eqref{A alpha beta}-\eqref{A loss}.

\begin{definition}[Matrix with sources and sinks]
Let $A \in \mathbb R^{L \times L } $ be a Markovian matrix.
Assume that the matrix $A$ is such that there exist a partition $\{ \alpha, \alpha_{in}, \alpha_{out} \} $ of $\Omega$, such that
\begin{equation} \label{A out in} 
 A_{\alpha , \alpha_{in} }\neq \textbf{0}, \quad  A_{ \alpha_{in}, \alpha  } = \textbf{0},  \quad
A_{\alpha_{out}, \alpha } \neq \textbf{0}, \quad  A_{ \alpha , \alpha_{out}} = \textbf{0}
 \quad
A_{\alpha_{out}, \alpha_{in} } = \textbf{0}, \quad  A_{ \alpha_{in} , \alpha_{out}} = \textbf{0}.
\end{equation}
Finally, we assume that the dynamics in each of the compartments is Markovian. 
Hence we assume that $ e_{|\alpha |}^T E_{\alpha }=0 $,  $e_{|\alpha_{in} |}^T E_{\alpha_{in} }=0$ and  $e_{|\alpha_{out} |}^T E_{\alpha_{out} }=0$. Then, we call $\alpha_{in}$ the source and $\alpha_{out}$ the sink of the chemical network.
\end{definition}

Notice that, if $A$ is a matrix with sources and sinks, then the assumptions \eqref{A out in} on $A$ imply that
\[ 
C_\alpha = \operatorname{diag}\left( \left\{ \sum_{j \in \alpha_{out}} A^i_j \right\}_{i\in\alpha} \right), \quad  C_{\alpha_{in} }= \operatorname{diag}\left( \left\{ \sum_{j \in \alpha} A^i_j \right\}_{i\in\alpha_{in}} \right) \text{ and } C_{\alpha_{out} }=\textbf{0}. 
\]

Using the notations introduced above and the assumptions on $A$ and on the set of compartments $X=\{\alpha , \alpha_{in}, \alpha_{out}\} $ we can rewrite \eqref{eq:ODE} as the following equation for $n(t)=(n_\alpha , n_{\alpha_{in}} , n_{\alpha_{out}})$, 
 \begin{align*}
 \frac{dn_\alpha (t)  }{dt } &= E_{\alpha} n_\alpha (t)  - C_\alpha n_\alpha (t) + A_{\alpha , \alpha_{in} } n_{\alpha_{in} } (t)  \\
  \frac{dn_{\alpha_{in}} (t)   }{dt } &= E_{\alpha_{in}} n_{\alpha_{in}}(t)  - C_{\alpha_{in}} n_{\alpha_{in}}(t)    \\
 \frac{dn_{\alpha_{out}} (t)  }{dt } &= E_{\alpha_{out}} n_{\alpha_{out}} (t) + A_{\alpha_{out}, \alpha } n_{\alpha } (t).  
  \end{align*}
 We refer to \cite{franco2023description} for more details on the reduction of systems of ODEs of the form \eqref{eq:ODE} in system of equations for the density of elements in a compartment.

\begin{definition}[Extended detailed balance]\label{def:EDB}
Assume that $A\in \mathbb R^{L \times L} $ is a Markovian matrix with sources and sinks. 
 We say that the matrix $A$ satisfies the extended detailed balance condition if the matrix $E_\alpha $ satisfies the detailed balance condition. Namely, there exists an $N \in \mathbb R_+^{|\alpha|} $ such that $E_\alpha N=0$ and such that the matrix $B_\alpha := S_\alpha^{-1} E_\alpha S_\alpha $, where  
 \[
 S_\alpha:=\operatorname{diag} \left( \{ \sqrt{N_j} \}_{j \in \alpha} \right)
 \]
is symmetric.
\end{definition}

In particular notice that if the ODE \eqref{eq:ODE} satisfies the detailed balance condition then $E_\alpha  =A $ and $n_\alpha=n$.

When a matrix $A$ has sources and sinks then we define for $i,j \in \alpha  $
\begin{equation}\label{response functions sources and sinks}
R_{ij} (t) := 
\langle e_j , e^{ t \left( E_\alpha - C_\alpha  \right)  } e_i \rangle \text{ for } i, j \in \alpha 
\end{equation}
as the response function of $j$ to a signal $i$. 
We refer to \cite{franco2023description} for more details on response functions.

\begin{theorem} \label{thm:EDB implies R12=R21}
    Assume that the matrix $A$ has sources and sinks and it satisfies the extended detailed balance condition for $\alpha \subset \Omega $. Assume $1, 2 \in \alpha $.
    Then 
    \begin{equation} \label{reversibility response} 
    R_{21}(t)= \frac{N_1}{N_2}R_{12}(t) \quad \text{ for all } t \geq 0 
    \end{equation}
    where the response functions $R_{12}, R_{21} $ are given by \eqref{response functions sources and sinks} and $N \in \mathbb R_+^{|\alpha|} $ is such that $E_\alpha  N=0$. 
\end{theorem}
\begin{proof}
    The extended detailed balance condition implies that 
    \begin{align*}
    R_{21}(t)&=\langle e_1, e^{ t \left( E_\alpha - C_\alpha \right)  } e_2 \rangle = \langle e_1, e^{ t S_\alpha \left( B_\alpha- S_\alpha^{-1}  C_\alpha S_\alpha \right) S_\alpha^{-1} } e_2 \rangle = \langle e_1, e^{ t S_\alpha \left( B_\alpha- S_\alpha^{-1}  C_\alpha S_\alpha \right) S_\alpha^{-1} } e_2 \rangle \\
    &=  \langle S_\alpha e_1  , e^{ t \left( B_\alpha- S_\alpha^{-1}  C_\alpha S_\alpha \right)  } S_\alpha^{-1}e_2 \rangle =\frac{\sqrt{N_1}}{\sqrt{N_2}} \langle e_1  , e^{ t \left( B_\alpha- S_\alpha^{-1}  C_\alpha S_\alpha \right)  } e_2 \rangle \\
    &= \frac{\sqrt{N_1}}{\sqrt{N_2}} \langle  e^{ t \left( B_\alpha- S_\alpha^{-1}  C_\alpha S_\alpha \right)^T } e_1, e_2 \rangle = \frac{N_1}{N_2} \frac{\sqrt{N_2}}{\sqrt{N_1}} \langle  e_2, e^{ t \left( B_\alpha- S_\alpha^{-1}  C_\alpha S_\alpha \right)  } e_1 \rangle \\
    &= \frac{N_1}{N_2} \langle  e_2, e^{ t S_\alpha \left( B_\alpha- S_\alpha^{-1}  C_\alpha S_\alpha \right)  S_\alpha^{-1}} e_1 \rangle =  \frac{N_1}{N_2}  \langle e_2, e^{ t \left( E_\alpha - C_\alpha \right)  } e_1 \rangle = \frac{N_1}{N_2} R_{12}(t).
    \end{align*}
\end{proof}

\section{Alternative choices of measurements}
\label{sec:non reverse measurements}
In the previous section we study if the detailed balance condition of a matrix $A$ can be detected by reciprocal measurements $R_{12} $, $R_{21} $ or non-reciprocal measurements. 
In this section we explain why we decided to focus our attention only on these reverse measurements.
More precisely, we consider the set of response functions $R_{12}, R_{13}$ corresponding to matrices $A \in \mathbb R^{ 4 \times 4 } $ that corresponds to a complete graph $\mathcal G(A) $.
We will show that, the fact that $A$ satisfies the detailed balance condition does not lead to any restriction on the response functions $R_{12}, R_{13} $. Hence this type of measurements are not useful in order to detect if the matrix $A$ satisfies or not detailed balance. 

Before doing that we write a heuristic computation that give some insights on the dimensionality of different sets of matrices. For instance, we consider the set of matrices satisfying the detailed balance condition respectively the pathwise detailed balance condition. 
These heuristic computations are in agreement with the results in the previous sections and allow us to see the notion of stability from a "geometrical perspective". For the purposes of this section we assume that the manifolds of matrices satisfying the detailed balance property respectively the pathwise detailed balance condition are well defined.

Consider a matrix $A \in \mathcal A(L) $ corresponding to a complete graph $\mathcal G(A)$. We recall that $\mathcal A $ is the set of Markovian and ergodic matrices.

We use the notation
\[
\mathcal B := \{  A \in \mathcal A : A \text{ satisfies the detailed balance condition} \} 
\] 
Finally, we denote with 
\[ 
\mathcal C  :=  \{  A \in \mathcal A : A \text{ satisfies the pathwise detailed balance condition} \} 
\] 
Notice that since the matrices in $\mathcal A $ are Markovian we have that 
\[
\dim (\mathcal A)= L^2 - L =L(L-1). 
\]

We now compute the dimension of the manifold $\mathcal B $. 
Notice that for every edge $e$ we can chose freely a reaction rate, while the reaction rate of the reverse edge  $e^*$ is determined by the steady state and by the reaction rate of $e$.
Therefore, the dimension of $\mathcal B$ is given by the number of edges in the graph, that is $\frac{L(L-1)}{2}$, plus the possible choices for the steady state, that is $L-1$. In total we have
\[ 
\dim(\mathcal B ) = \frac{L(L-1)}{2}+L-1 =\frac{1}{2} (L-1)(L+2). 
\]
Finally $\dim(\mathcal C ) $ is $\dim (\mathcal A ) $ minus the $L-1$ constraints (assuming that they are all independent) of the pathwise detailed balance condition $+1$, due to the choice of the constant $c$ in the pathwise detailed balance condition
\[
 \dim(\mathcal C)=  L (L-1) - (L-1) +1  =L (L-2) + 2. 
 \] 
Notice that for $L=3 $ we have $\dim (\mathcal B )= \dim (\mathcal C) =5 $ as expected from Corollary \ref{cor:L=3}, where the equivalence between pathwise detailed balance and detailed balance is proven. 
Instead, for $L=4 $ we have $9= \dim (\mathcal B ) <  \dim (\mathcal C) =10$, as expected from Example \ref{exam:L=4 strongly connected}, where it is proven that pathwise detailed balance is not equivalent to detailed balance. 

Notice that for generic $L > 3 $ we have $\dim (\mathcal A ) > \dim (\mathcal C) > \dim (\mathcal B) $.
As a consequence of the fact that $\dim (\mathcal A ) > \dim (\mathcal C) $ we can expect that, given a matrix $A\in \mathcal C $, every open set $U \subset \mathcal A $ with $A \in U$ satisfies $U \cap (\mathcal A \setminus \mathcal C) \neq \emptyset$. 
In other words, this geometrical argument suggests that the pathwise detailed balance condition is unstable in the class of complete graphs with at least one edge without detailed balance. 
This is exactly what is proven in Corollary \ref{cor:unstable PDB complete graph with one detailed balance}.

Moreover, let $\mathcal C_d $ be the submanifold of $\mathcal A $ for which pathwise detailed balance holds for $ d $ couple of points. 
Then we compute $\dim(\mathcal C_d) $ as we did before with $\mathcal C $, but taking into account that now we have $d (L-1)$ constraints due to the pathwise detailed balance condition and we have $d $ different constants. Here we assume that all $d$ measurements give $d$ independent constraints. Hence,
\[
\dim (\mathcal C_d ) = L(L-1) - d (L-1) + d. 
\]
Assume now that $L $ tends to infinity. We want to understand how many couples of reverse measurements (hence of the form $R_{ij}, R_{ji}$) we should do in order to have $\dim (\mathcal C_d) = \dim (\mathcal B) $.
It turns out that 
\[
d =  \left(  L^2 / 2 - 3/2 L +1 \right)/(L-2).
\]
Hence as $L \rightarrow \infty $ we have $d \approx L /2 $. This implies that, in order to be able to prove the detailed balance condition, performing reverse measurements we need to make a number of measurements of the order of $L$.
Therefore, from the practical point of view this approach is not feasible. This is why in this paper we focus instead on studying the stability of the pathwise detailed balance condition for some classes of matrices. 

Finally, consider the Markovian matrices $A \in \mathbb R^{L \times L } $ such that $\mathcal G(A)$ is a cycle graph.
We denote this set of matrices as $\mathcal A_c$. 
Notice that, these matrices (as well as their corresponding steady state) are determined by the weights of the edges along the cycle, hence 
\[
\dim (\mathcal A_c) = 2L.
\]
Notice that, under these assumptions, the pathwise detailed balance condition provides only two constraint, indeed \eqref{eq:PDB} is non trivial only for $n=1$ and $n=L-1$. 
Hence, including the degree of freedom given by the steady state, we get
\[ 
\dim \left( \{  A \in \mathcal A_c : A \text{ satisfies the pathwise detailed balance condition} \} \right)=2L -2 +1=2L -1 . 
\]
Finally the dimension of the set of the matrices in $\mathcal A_c $ that satisfy the detailed balance condition can be computed as above, namely it is the sum of the number of edges plus the number of degree of freedom provided by the steady state, namely
\[
\dim \left( \{  A \in \mathcal A_c : A \text{ satisfies the detailed balance condition} \} \right)= L + L-1 =  2L -1 . 
\]
Therefore, the dimension of the set of matrices in $\mathcal A_c $ that satisfy the detailed balance condition is equal to the dimension of the set of matrices in $\mathcal A_c $ that satisfy the pathwise detailed balance condition. This is in agreement with the equivalence of the two properties (detailed balance and pathwise detailed balance), proven in Lemma \ref{lem:p-DB iff DB on cycles}.

We  now consider different sets of possible response functions. We assume that $L= 4 $ and that the set of measurements is $\mathcal M :=\{ (1,2), (1,3) \} \subset \{1, \dots, 4\}^2  $, i.e.~we consider the response functions 
\[
R_{12} (t)= \langle e_2 , e^{ t A } e_1 \rangle \text{ and } R_{13} (t)= \langle e_3 , e^{ t A } e_1 \rangle . 
\]
Since we are considering $L=4 $ the response function $R_{12} $ is characterized by the constants $\langle e_2 , A  e_1 \rangle$, $\langle e_2$ , $A^2  e_1 \rangle, \langle e_2 , A^3  e_1 \rangle $ while the response function $R_{13} $ is characterized by the constants $\langle e_3 , A  e_1 \rangle$, $\langle e_3$ , $A^2  e_1 \rangle, \langle e_3 , A^3  e_1 \rangle $. 

Let $F : \mathcal A \rightarrow \mathbb R^{6} $ be given by
\[
F : A \mapsto (  \langle e_2 , A  e_1 \rangle, \langle e_2 , A^2  e_1 \rangle, \langle e_2 , A^3  e_1 \rangle \langle e_3 , A  e_1 \rangle, \langle e_3 , A^2  e_1 \rangle, \langle e_3 , A^3  e_1 \rangle )
\]
Let 
\[
A_0 = \left( 
\begin{matrix}
    &-3&1 &1 &1 \\
    & 1 & -3 &1 &1 \\
    & 1 & 1 & -3 &1 \\
    &1 &1 &1 &-3
\end{matrix}\right).
\] 
We want to show that $F(B_\delta(A_0))$, with $B_\delta(A_0) \subset \mathcal A_{\emptyset,\{ 1, \dots, 4\}^2 } $ is an open set in $\mathbb R^6$. 
To this end we show that the linearization of $F$ around $A_0 $ is of full rank.  
Notice that 
\[
DF(A_0) (\xi) = (x_1,x_2, x_3, x_4,x_5,x_6)
\] 
where $\xi = (\xi_{12}, \xi_{13}, \xi_{14}, \xi_{23}, \xi_{24}, \xi_{34})$ is the vector of the values that characterize a $4 \times 4 $ symmetric matrix in $\mathcal A $,  and where 
\begin{align*}
x_1&= \xi_{12}, \\
x_2&=\xi_{13}+\xi_{32} + \xi_{14}+\xi_{42}, \\
 x_3&=4 \xi_{13} +7\xi_{12}+4\xi_{14} +2\xi_{34}+2\xi_{42} +2\xi_{32},\\
 x_4&= \xi_{13}, \\
 x_5 &=\xi_{12} +\xi_{23}+\xi_{14}+\xi_{43} \\
 x_6 &= 7 \xi_{13}+3\xi_{12} +3\xi_{14} +3\xi_{34}+2 \xi_{42}+3\xi_{32}.
\end{align*}

Then 
\[
DF(A_0)(\xi) = M \xi
\] 
where 
\[
M := \left( \begin{matrix}
   & 1 & 0 & 0 &0 & 0 & 0 \\
   & 0 & 1 &1 &1 &1 &0 \\
   & 7 & 4 & 4 & 2 &2 &2 \\
   & 0 & 1 & 0 & 0 &0 & 0  \\
   & 1 & 0  &1 & 1 & 0 &1 \\
   & 3 & 7 & 3 &3 &2 &3 
\end{matrix}\right)
\]
Since $\det M \neq 0 $ the desired conclusion follows: we can cover an open set of coefficients of the response functions $R_{12}$ and $R_{13}$ with matrices that satisfy the detailed balance condition.
Hence the fact that a matrix $A$ is symmetric (or satisfies the detailed balance condition) does not impose any condition on the response functions $R_{12}, R_{13}$. It is therefore convenient to consider the relation between the response functions  $R_{12} $ and $R_{21}$, as we do in this paper. Indeed in that case we have that the detailed balance condition of $A$ imposes strong conditions on the relation between $R_{12}$ and $R_{21}$. See Proposition \ref{lem:DB implies pDB} and Proposition \ref{prop:R12=R21 implies PSEUDO DB}.

\section{Concluding remarks} \label{sec:conclusion}
We conclude summarizing the main results that we obtained in this paper and indicate possible continuations of this work. 
The aim of this paper is to understand whether a biochemical network satisfies the detailed balance condition by performing a couple of reciprocal measurements ($R_{ij}$ and $R_{ji}$). The advantage of this approach is that it does not require a detailed knowledge of the rates of the reactions of the network. 
We proved that the pathwise detailed balance condition is a necessary condition for detailed balance to hold. 
We studied the stability property of the pathwise detailed balance condition in suitably defined classes of biochemical networks. Interestingly, it turns out that the stability of the pathwise detailed balance property is strictly related with the topological properties of the network, in particular to the existence of cut vertices. 
In particular, we proved that if a network does not contain a cut vertex and satisfies the pathwise detailed balance in a stable way, then it satisfies the detailed balance condition. 

The definition of stability introduced in this paper provides a possible rigorous definition of robustness for biochemical networks. 
In this paper we refer essentially always to the stability of the pathwise detailed balance condition, but similar definitions of stability/robustness could also be extended to other contexts. This definition could be modified in order to adapt to the knowledge available about the evolution of biochemical networks (see \cite{noda2018metabolite}). 

Finally, it would be interesting to generalize the results of this paper to non-linear reaction networks. The challenge in this case is that one cannot rely on a straightforward connection with graph theory, as we did in this paper. 
It would also be interesting to analyse the case in which the set of states $\Omega $ is not discrete. For instance, it is relevant to determine if the  molecular concentrations in a cell are due to passive transport processes, like diffusion, or active processes, for instance the action of molecular motors.

\bigskip 

\textbf{Acknowledgements}
The authors gratefully acknowledge the support by the Deutsche Forschungsgemeinschaft (DFG) through the collaborative research centre "The mathematics of emerging
effects" (CRC 1060, Project-ID 211504053). B. Kepka is funded by the Bonn International Graduate School of Mathematics at the Hausdorff Center for Mathematics
(EXC 2047/1, Project-ID 390685813). E. Franco and J. J. L. Velázquez are funded by the DFG
under Germany's Excellence Strategy-EXC2047/1-390685813. The funders had no role in study design, analysis, decision to
publish, or preparation of the manuscript.

\bibliographystyle{habbrv}
\bibliography{References}

 \end{document}